\documentclass[sigconf]{aamas}

\usepackage{balance} 

\setcopyright{ifaamas}
\acmConference[AAMAS '26]{Proc.\@ of the 25th International Conference
on Autonomous Agents and Multiagent Systems (AAMAS 2026)}{May 25 -- 29, 2026}
{Paphos, Cyprus}{C.~Amato, L.~Dennis, V.~Mascardi, J.~Thangarajah (eds.)}
\copyrightyear{2026}
\acmYear{2026}
\acmDOI{}
\acmPrice{}
\acmISBN{}

\acmSubmissionID{282}

\title[AAMAS-2026 Formatting Instructions]{Networked Communication for Mean-Field Games with Function Approximation and Empirical Mean-Field Estimation}

\author{Patrick Benjamin}
\affiliation{
  \institution{University of Oxford}
  \city{Oxford}
  \country{United Kingdom}}
\email{phlbenjamin@gmail.com}

\author{Alessandro Abate}
\affiliation{
  \institution{University of Oxford}
  \city{Oxford}
  \country{United Kingdom}}
\email{alessandro.abate@cs.ox.ac.uk}

\begin{abstract}
Recent algorithms allow decentralised agents, possibly connected via a communication network, to learn equilibria in mean-field games from a non-episodic run of the empirical system. However, these algorithms are for tabular settings: this computationally limits the size of agents’ observation space, meaning the algorithms cannot handle anything but small state spaces, nor generalise beyond policies depending only on the agent’s local state to so-called ‘population-dependent’ policies. We address this limitation by introducing function approximation to the existing setting, drawing on the Munchausen Online Mirror Descent method that has previously been employed only in finite-horizon, episodic, centralised settings. While this permits us to include the mean field in the observation for players’ policies, it is unrealistic to assume decentralised agents have access to this global information: we therefore also provide new algorithms allowing agents to locally estimate the global empirical distribution, and to improve this estimate via inter-agent communication. We prove theoretically that exchanging policy information helps networked agents outperform both independent and even centralised agents in function-approximation settings. Our experiments demonstrate this happening empirically, and show that the communication network allows decentralised agents to estimate the mean field for population-dependent policies.
\end{abstract}


\newcommand{\BibTeX}{\rm B\kern-.05em{\sc i\kern-.025em b}\kern-.08em\TeX}

\usepackage{amsthm}

\theoremstyle{plain}
\newtheorem{theorem}{Theorem}[section]

\theoremstyle{definition}  
\newtheorem{definition}[theorem]{Definition}
\newtheorem{assumption}[theorem]{Assumption}
\newtheorem{remark}[theorem]{Remark}

\usepackage{algorithm}
\usepackage{algorithmic}
\usepackage{graphicx}
\usepackage{booktabs}
\usepackage{dsfont}
\usepackage{subcaption}
\usepackage{float} 
\begin{document}


\pagestyle{fancy}
\fancyhead{}


\maketitle 

\section{Introduction}\label{introduction}

The mean-field game (MFG) framework \citep{lasry2007mean,huangMFG} can be used to circumvent the difficulty faced by multi-agent reinforcement learning regarding computational scalability as the number of agents grows \citep{yardim2024exploiting,zeng2024single}.\footnote{Preliminary versions of this paper have been accepted at: MARW at AAAI'25; GAAMAL at IJCAI'25; MALTA at AAAI'25; ALA at AAMAS’25; EMAS at AAMAS’25; GAIW at AAMAS’25; OptLearnMAS-25 at AAMAS’25; RLDM’25.} It models a representative agent as interacting not with other individual agents in the population on a per-agent basis, but instead with a distribution over the other agents, called the \textit{mean field}. The MFG framework analyses the limiting case when the population consists of an infinite number of symmetric and anonymous agents, that is, they have identical reward and transition functions which depend on the mean-field distribution rather than on the actions of specific other players. The solution to this game is the mean-field Nash equilibrium (MFNE), which can be used as an approximation for the Nash equilibrium (NE) in a finite-agent game (which is harder to solve in itself), with the error in the solution reducing as the number of agents \textit{N} tends to infinity \citep{finite_proofs,Anahtarci2020QLearningIR,yardim2024mean,TOUMI2024111420,hu2024mfoml,chen2024periodic}.  MFGs have thus been applied to a wide range of real-world problems: see \citet{survey_learningMFGs} for examples.


Recent works argue that classical algorithms for solving MFGs rely on assumptions and methods that are likely to be undesirable in real-world applications (e.g. swarm robotics, autonomous vehicles), emphasising that desirable qualities for practical MFG algorithms include: learning from the empirical distribution of $N$ agents (i.e. this distribution is generated only by the policies of the agents, rather than being updated by the algorithm itself or an external oracle/simulator); learning online from a single, non-episodic system run (i.e. similar to above, the population cannot be arbitrarily reset by an external controller); model-free learning; decentralised learning; and fast practical convergence \citep{policy_mirror_independent,benjamin2024networked}. While these works address these desiderata, they do so only in settings in which the state and action spaces are small enough that the Q-function can be represented by a table, limiting their approaches' scalability. 

Moreover, in those works, as in many others on MFGs, agents only observe their local state as input to their Q-function (which defines their policy). This is sufficient when the solved MFG is expected to have a stationary distribution (`stationary MFGs') \citep{survey_learningMFGs,one_that_may_sample,Anahtarci2020QLearningIR,zaman2023oraclefree,policy_mirror_independent,benjamin2024networked}. However, 
in reality there are numerous reasons why agents may benefit from being able to respond to the current distribution (discussed further in Appx. \ref{Further_discussion_non_stationary_equilibria}). Recent work has thus increasingly focused on these more general settings where it is necessary for agents to have so-called `master policies' (a.k.a. population-dependent policies) which depend on both the mean-field distribution and their local state \citep{cardaliaguet2015masterequationconvergenceproblem,carmona_common,perrin2022generalization,wu2024populationaware,survey_learningMFGs,scalable_deep}.

The distribution is a large, high-dimensional observation object, taking a continuum of values. Therefore a population-dependent Q-function cannot be represented in a table and must be approximated. To address these limitations while maintaining the desiderata for real-world applications given in recent works, we introduce function approximation to the MFG setting of decentralised agents learning online from a single, non-episodic run of the empirical system, allowing this setting to handle larger state spaces and to accept the mean-field distribution as an observation input. 
To overcome the difficulties of training non-linear approximators in this context, we use the so-called `Munchausen' trick, introduced first for single-agent RL \citep{NEURIPS2020_2c6a0bae}, extended to MFGs \citep{scalable_deep}, and then to MFGs with population-dependent policies \citep{wu2024populationaware}.

We particularly explore this in the context of networked communication between decentralised agents. Almost all prior work relies on a centralised node to learn on behalf of all the agents. In this context `centralised' does not imply global observability of the whole population's actions - which would generally make computation infeasible given the complexity of the problem - but rather that learning is only conducted from the samples of a single representative agent, whose policy updates are assumed to be automatically pushed to the rest of the population by the central node. Therefore to reduce confusion, we sometimes refer to `central-agent learning' instead of `centralised learning' in contrast to prior works. More recent works have recognised that the assumption of a central learner might be be unrealistic in the real world, as well as a computational bottleneck and a vulnerable single point of failure \citep{benjamin2024networked, policy_mirror_independent}. 

We demonstrate that networked communication brings two specific benefits over the purely independent setting, while also removing the undesirable assumption of a central learner. Firstly, when the Q-function is approximated rather than exact, some updates lead to better performing policies than others. Allowing networked agents to propagate better performing policies through the population leads to faster learning than in the purely independent case and very often even than in the central-agent case, as we show both theoretically and empirically (this method is reminiscent of the use of fitness functions in distributed evolutionary algorithms \citep{Eiben2015,survivability}, and similarly of `population-based training' \citep{jaderberg2017population}). Secondly, we argue that in the real world it is unrealistic to assume that decentralised agents, endowed with local state observations and limited (if any) communication radius, would be able to observe the global mean-field distribution and use it as input to their Q-functions / policy. We therefore further contribute two setting-dependent algorithms by which decentralised agents can estimate the global distribution from local observations, and further improve their estimates by communication with neighbours. 

We focus on `coordination' games, where agents can increase their individual rewards by following the same strategy as others and therefore have an incentive to communicate policies, even if the MFG setting itself is technically non-cooperative. Thus our work can be applied to real-world problems in e.g. traffic signal control, robotic swarm formation control, vehicle platooning, and consensus and synchronisation e.g. for sensor networks \citep{10757965}.\footnote{We further preempt concerns about 
communication in competitive settings by wondering whether self-interested agents would be any more likely to want to obey a central learner as has usually been assumed. Moreover we show that self-interested communicating agents can obtain higher returns than independent agents even in non-coordination games (Fig. \ref{diffuse100}), indicating that they do have incentive to communicate.} In summary, our contributions are:

\begin{itemize}
    \item We introduce, for the first time, function approximation to MFG settings with decentralised agents
    . To do this: 
    \begin{itemize}
        \item We use Munchausen RL for the first time in an infinite-horizon MFG context 
        (cf. finite-horizon \citep{scalable_deep,wu2024populationaware}).
        \item This constitutes the first use of function approximation for solving MFGs from a single, non-episodic run of the empirical system 
        (cf. tabular settings \citep{policy_mirror_independent,benjamin2024networked}).
    \end{itemize}
    \item Function approximation allows us to explore larger state spaces, and also settings where agents' policies depend on the mean-field distribution as well as their local state. 
    \item Instead of assuming that agents have access to this global knowledge as in prior works, we present two additional novel algorithms allowing decentralised agents to locally estimate the empirical distribution 
    and to improve these estimates by inter-agent communication.
    \item We prove theoretically that networked agents can learn faster than even central-agent populations in the function-approximation setting. 
    \item We support this with extensive experiments, where our results demonstrate the two benefits of {the decentralised communication scheme, which significantly outperforms both the independent and central-agent architectures}.
\end{itemize}

The paper is structured as follows. Related work is given in Sec. \ref{related_work_section}. We give preliminaries in Sec. \ref{preliminaries_section} and our core learning and policy-improvement algorithm in Sec. \ref{learning_and_policy_improvement}. We present our mean-field estimation and communication algorithms in Sec. \ref{estimation}, theoretical results in Sec. \ref{theoretical_results} and experiments in Sec. \ref{experiments_section}.

\section{Related work}\label{related_work_section}
We refer the reader to \citet{benjamin2024networked} for detailed discussion around the setting of networked communication in MFGs, and to \citet{survey_learningMFGs} for a broader survey of MFGs. Our work is most closely related to \citet{benjamin2024networked}, which introduced networked communication to the infinite-horizon MFG setting
. However, this work focuses only on  tabular settings rather than using function approximation as in ours, and only addresses population-independent policies. 

\citet{scalable_deep} uses Munchausen Online Mirror Descent (MOMD), similar to our method for learning with neural networks, but in a different setting: 
their mean-field distribution is updated in an exact way and an oracle supplies a central learner with rewards and transitions for it to learn a population-independent policy, in a finite-horizon, episodic setting. \citet{wu2024populationaware} uses MOMD to learn population-dependent policies, albeit also with a central-agent method that exactly updates the mean-field distribution in a finite-horizon episodic setting. \citet{perrin2022generalization} learns population-dependent policies with function approximation in infinite-horizon settings like our own, but does so in a central-agent, two-timescale manner without using the empirical mean-field distribution. \citet{zhang2024stochastic} also uses function approximation along a non-episodic path, but involves a generic central agent learning an estimate of the mean field rather than using an empirical population. Approaches that directly update an estimate of the mean field must be able to generate rewards from this arbitrary mean field, even if they otherwise claim to be oracle-free. They are thus inherently centralised algorithms and rely on strong assumptions that may not apply in real-world problems. Conversely, we are interested in practical convergence in online, deployed settings, where the reward is computed from the empirical finite population.

\citet{subjective_equilibria} addresses decentralised learning from a continuous, non-episodic run of the empirical system using either full or compressed information about the mean field, but agents are assumed to receive this information directly, rather than estimating it locally as in the algorithm we now present. They also do not consider function approximation or inter-agent communication
. In the closely related but distinct area of mean-field RL, \citet{pomfrl} does estimate the empirical mean-field distribution from the local neighbourhood, however agents are seeking to estimate the mean action rather than the mean-field distribution over states as in our MFG setting. Their agents also do not have access to a communication network by which they can improve their estimates. 


\section{Preliminaries}\label{preliminaries_section}

\subsection{Mean-field games}

We use the following notation. $N$ is the number of agents in a population, with $\mathcal{S}$ and $\mathcal{A}$ representing the finite state and common action spaces, respectively. The set of probability measures on a finite set $\mathcal{X}$ is denoted $\Delta_\mathcal{X}$, and $\mathbf{e}_x \in \Delta_\mathcal{X}$ for $x \in \mathcal{X}$ is a one-hot vector with only the entry corresponding to $x$ set to 1, and all others set to 0. For time $t \geq 0$, $\hat{\mu}_t$ = $\frac{1}{N}\sum^N_{i=1}\sum_{s\in\mathcal{S}}$ $\mathds{1}_{s^i_t=s}\mathbf{e}_s$ $\in$ $\Delta_\mathcal{S}$ is a vector  of length $|\mathcal{S}|$ denoting the empirical categorical state distribution of the $N$ agents at time $t$. For agent $i\in1\dots N$, $i$'s policy at time $t$ depends on its observation $o^i_t$. We explore three different forms that this observation object can take:

\begin{itemize}
    \item In the conventional setting, the observation is simply $i$'s current local state $s^i_t$, such that $\pi^i(a|o^i_t) = \pi^i(a|s^i_t)$.
    \item When the policy is population-dependent, if we assume perfect observability of the global mean-field distribution then we have $o^i_t = (s^i_t, \hat{\mu}_t)$. 
    \item It is unrealistic to assume that decentralised agents with a possibly limited communication radius can observe the global mean field, so we allow agents to form a local estimate $\tilde{\hat{\mu}}^i_t$ which can be improved by communication with neighbours. Here we have $o^i_t = (s^i_t, \tilde{\hat{\mu}}^i_t)$. 
\end{itemize}

In the following definitions we focus on the population-dependent case when $o^i_t = (s^i_t, \hat{\mu}_t)$, and {clarify afterwards the connection to the other observation cases}. Thus the set of policies is $\Pi$ = \{$\pi$ : $\mathcal{S} \times \Delta_\mathcal{S} \rightarrow \Delta_\mathcal{A}$\}, {and the set of Q-functions is denoted $\mathcal{Q} = \{q : \mathcal{S} \times \Delta_\mathcal{S} \times \mathcal{A} \rightarrow \mathbb{R}\}$}.   

\begin{definition}[\textit{N}-player symmetric anonymous games] 
An N-player stochastic game with symmetric, anonymous agents is given by the tuple $\langle$$N$, $\mathcal{S}$, $\mathcal{A}$, $P$, $R$, $\gamma$$\rangle$, where $\mathcal{A}$ is the action space, identical  for each agent; $\mathcal{S}$ is the identical  state space of each agent, such that 
{their initial states are \{$s^i_0$\}$_{i=1}^N \in \mathcal{S}^N$} and their policies are \{$\pi^i$\}$_{i=1}^N \in \Pi^N$. $P$ : $\mathcal{S}$ $\times$ $\mathcal{A}$ $\times$ $\Delta_{\mathcal{S}}$ $\rightarrow$ $\Delta_{\mathcal{S}}$ is the transition function and $R$ : $\mathcal{S}$ $\times$ $\mathcal{A}$ $\times$ $\Delta_{\mathcal{S}}$ $\rightarrow$ [0,1] is the reward function,  which map each agent's local state and action and the population's empirical distribution to transition probabilities and bounded rewards, respectively, i.e.:
\[s^i_{t+1} \sim P(\cdot|s^i_{t},a^i_{t},\hat{\mu}_t), \;\;\; r^i_{t} = R(s^i_{t},a^i_{t},\hat{\mu}_t)\;\;\;\;\;\;\;\;\;\forall i = 1,\dots,N.\] 
\end{definition}

At the limit as $N \rightarrow \infty$, the infinite population of agents can be characterised as a limit distribution $\mu \in \Delta_\mathcal{S}$; the infinite-agent game is termed an MFG. The so-called `mean-field flow' $\boldsymbol\mu$ is given by the infinite sequence of mean-field distributions s.t.  $\boldsymbol\mu = $$(\mu_t)_{t\geq0}$.

\begin{definition}[Induced mean-field flow] 
We denote by $I(\pi)$ the mean-field flow $\boldsymbol\mu$ induced when all the agents follow $\pi$, where this is generated from $\pi$ as follows:
\[\mu_{t+1} (s') = \sum_{s,a}\mu_t (s)\pi (a|s,\mu_t)P(s'|s,a,\mu_t).\] \end{definition}

When the mean-field flow is stationary such that the distribution is the same for all $t$, i.e. $\mu_t=\mu_{t+1}$ $\forall t \geq 0$, the policy $\pi^i(a|s^i_t, 
{\mu}_t)$ need not depend on the distribution, such that $\pi^i(a|s^i_t, 
{\mu}_t)$ = $\pi^i(a|s^i_t)$, i.e. we recover the classical population-independent policy. 
However, for such a population-independent policy the initial distribution $\mu_0$ must be known and fixed in advance, whereas otherwise it need not be. We also give the following definitions.

\begin{definition}[Mean-field discounted return]\label{Mean_field_discounted_return}
In a MFG where all agents follow policy $\pi$ giving a mean-field flow $\boldsymbol\mu = $$(\mu_t)_{t\geq0}$, the expected discounted return of the representative agent is given by
\[V(\pi,\boldsymbol\mu) = \mathbb{E}\left[
\sum^{\infty}_{t=0}\gamma^{t}(R(s_t,a_t,\mu_t))\bigg|
\substack{s_0\sim\mu_0\\ 
a_t\sim \pi(\cdot|s_t,\mu_t)\\ s_{t+1} \sim P(\cdot|s_{t},a_{t},{\mu_t})}\right]. 
\]
\end{definition}

\begin{definition}[Best-response (BR) policy] 
A policy $\pi^*$ is a \textit{best response (BR)} against the mean-field flow $\boldsymbol\mu$ if it maximises the discounted return  $V(\cdot,\boldsymbol\mu)$; the set of these policies is denoted $BR(\boldsymbol\mu)$: 
\[\pi^* \in BR(\boldsymbol\mu) := \mathop{\arg\max}\limits_{\pi}V(\pi,\boldsymbol\mu).\]
\end{definition}

\begin{definition}[Mean-field Nash equilibrium (MFNE)]
A pair ($\pi^*, \boldsymbol\mu^*)$ is a mean-field Nash equilibrium (MFNE) if the following two conditions hold:
\begin{itemize}
    \item $\pi^*$ is a best response to $\boldsymbol\mu^*$, i.e. $\pi^* \in BR(\boldsymbol\mu^*)$;
    \item $\boldsymbol\mu^*$ is induced by $\pi^*$, i.e.  $\boldsymbol\mu^*$ =  $I(\pi^*)$.
\end{itemize}
\end{definition}
$\pi^*$ is thus a fixed point of the map $BR\circ I$, i.e. $\pi^* \in BR(I(\pi^*))$. If a population-dependent policy is a MFNE policy for any initial distribution $\mu_0$, it is a `master policy'.

Previous works have shown that, in tabular settings, it is possible for a finite population of decentralised agents (each of which is permitted to have a distinct population-independent policy $\pi^i$) to learn the MFNE using only the empirical distribution $\hat{\mu}_t$, rather than the exactly calculated infinite flow $\boldsymbol\mu$ \citep{policy_mirror_independent,benjamin2024networked}. This MFNE may be the goal in itself, or it can in turn serve as an approximate NE for the harder-to-solve game involving the finite population. In this work we provide algorithms to perform this process in non-tabular and population-dependent settings, and demonstrate them empirically.

\subsection{(Munchausen) Online Mirror Descent}

Instead of finding a $BR$ at each iteration, which is computationally expensive, we can use a form of policy iteration for MFGs called Online Mirror Descent (OMD). This begins with an initial policy $\pi_0$, and then at each iteration $k$, evaluating the current policy $\pi_k$ with respect to its induced mean-field flow $\boldsymbol{\mu} = I(\pi_k)$ to compute its Q-function $Q_{k+1}$. To stabilise the learning process, we then use a weighted sum over this and past Q-functions, and set $\pi_{k+1}$ to be the softmax over this weighted sum, i.e.  $\pi_{k+1}(\cdot|s,\mu) = softmax\left(\frac{1}{\tau_q}\sum^k_{\kappa=0}{Q_\kappa}(s,\mu,\cdot)\right)$. $\tau_{q}$ is a temperature parameter that scales the entropy in Munchausen RL \citep{NEURIPS2020_2c6a0bae}; note that this is a different temperature to the one agents use when selecting which communicated parameters to adopt, denoted $\tau^{comm}_k$ (Sec. \ref{policy_adoption_and_communication}).

If the Q-function is approximated non-linearly using neural networks, it is difficult to compute this weighted sum. The so-called `Munchausen trick' addresses this by computing a single Q-function that mimics the weighted sum 
 using implicit regularisation based on the Kullback-Leibler (KL) divergence between $\pi_{k}$ and $\pi_{k+1}$ \citep{NEURIPS2020_2c6a0bae}. Using this reparametrisation gives Munchausen OMD (MOMD), detailed further in Sec. \ref{q_update} \citep{scalable_deep,wu2024populationaware}. MOMD does not bias the MFNE, and has the same convergence guarantees as OMD \citep{hadikhanloo2017learninganonymousnonatomicgames,perolat2021scaling,wu2024populationaware}.

\subsection{Networks}\label{preliminaries_networks}

We conceive of the finite population as exhibiting two time-varying networks
. The basic definition of such a network is:

\begin{definition}[Time-varying network] 
The network ($\mathcal{G}_t$)$_{t\geq0}$ is given by $\mathcal{G}_t$ = ($\mathcal{N}, \mathcal{E}_t$), where $\mathcal{N}$ is the set of vertices representing agents $i \in \{1,\dots,N\}$, and the edge set $\mathcal{E}_{t}$ $\subseteq$ \{(\textit{i},\textit{j}) : \textit{i},\textit{j} $\in$ $\mathcal{N}$, \textit{i} $\neq$ \textit{j}\} is the set of undirected links present at time \textit{t}. A network's \textit{diameter} $d_{\mathcal{G}_t}$ is the maximum of the shortest path length between any pair of nodes.
\end{definition}

One of these networks $\mathcal{G}^{comm}_t$ defines which agents can communicate information to each other at time $t$. 
The second network $\mathcal{G}^{obs}_t$ is a graph defining which agents can observe each other's states, which we use in general settings for estimating the mean-field distribution from local information. The structure of the two networks may be identical (e.g. if embodied agents can both observe the position (state) of, and exchange information with, other agents within a certain physical distance from themselves), or different (e.g. if agents can observe the positions of nearby agents, but only exchange information with agents by which they are linked via satellite, which may connect agents over long distances).

We also define an alternative version of the observation graph that is useful in a specific subclass of environments, which can most intuitively be thought of as those where agents' states are positions in physical space. When this is the case, we usually think of agents' ability to observe each other as depending more abstractly on whether states are visible to each other. This visibility graph is:

\begin{definition}[Time-varying state-visibility graph] 
The state visibility graph ($\mathcal{G}^{vis}_t$)$_{t\geq0}$ is given by $\mathcal{G}^{vis}_t$ = ($\mathcal{S'}, \mathcal{E}^{vis}_t$), where $\mathcal{S'}$ is the set of vertices representing the environment states $\mathcal{S}$, and the edge set $\mathcal{E}^{vis}_{t}$ $\subseteq$ \{(\textit{m},\textit{n}) : \textit{m},\textit{n} $\in$ $\mathcal{S'}$\} is the set of undirected links present at time \textit{t}, indicating which states are visible to each other.
\end{definition}

We say an agent in $s$ can obtain a count of the number of agents in $s'$ if $s'$ is visible to $s$. The benefit of this graph $\mathcal{G}^{vis}_t$ over $\mathcal{G}^{obs}_t$ is that there is mutual exclusivity: either an agent in state $s$ is able to obtain a total count of all of the agents in state $s'$ (if $s'$ is visible to $s$), or it cannot obtain information about any agent in state $s'$ (if those states are not visible to each other). Additionally, this graph permits an agent in state $s$ to observe that there are \textit{no} agents in state $s'$ as long as $s'$ is visible to $s$. These benefits are not available if the observability graph is defined strictly between agents as in $\mathcal{G}^{obs}_t$, such that using $\mathcal{G}^{vis}_t$ facilitates more efficient estimation of the global mean-field distribution from local information in settings where $\mathcal{G}^{vis}_t$ is applicable (see Sec. \ref{estimation}).


\begin{algorithm}[!tb]
    \caption{Networked learning with 
    function approximation} \label{networked_algorithm_neural_networks}
    \begin{algorithmic}[1]

    \REQUIRE loop parameters $K, M, L, E, C_p$;    learning parameters
    $\gamma, \tau_q, |B|, 
    cl, \nu$, $\{\tau^{comm}_k\}_{k \in \{0,\dots,K-1\}}$; initial states \{$s^i_0$\}$_{i=1}^{N}$; $t \leftarrow 0$ 
    \STATE $\forall i$ : Randomly initialise parameters $\theta^i_0$ of Q-networks $\check{Q}_{\theta^i_0}(o,\cdot)$, and set $\pi^i_0(a|o)= $ softmax$\left(\frac{1}{\tau_q}\check{Q}_{\theta^i_0}(o,\cdot)\right)(a)$\label{initalise_line}
    \FOR{$k = 0,\dots,K-1$}
        \STATE $\forall i$: Empty $i$'s buffer\label{empty_buffer}
        \FOR{$m = 0,\dots,M_{}-1$}\label{start_m_loop}
            \STATE Take step $\forall i : a^i_t \sim \pi^i_{k}(\cdot|o^i_t), r^i_{t} = R(s^i_{t},a^i_{t},\hat{\mu}_t),s^i_{t+1} \sim P(\cdot|s^i_{t},a^i_{t},\hat{\mu}_t)$; $t \leftarrow t + 1$\label{step1}
            \STATE $\forall i$: Add $\zeta^i_{t}$ to $i$'s buffer\label{store_transition}   
        \ENDFOR\label{end_m_loop}
        \FOR{$l = 0,\dots,L-1$}\label{beginlearningwithbuffer}
            \STATE $\forall i : $ Sample batch $B^i_{k,l}$ from $i$'s buffer\label{sample_from_buffer}
            \STATE Update $\theta$ to minimise $\hat{\mathcal{L}}(\theta,\theta')$ as in Def. \ref{loss}\label{minimise_loss_line}
            \STATE If $l\mod$ $\nu = 0$, set $\theta'\leftarrow\theta$ \label{update_target_line}
        \ENDFOR \label{endlearningwithbuffer}
        \STATE $\check{Q}_{\theta^i_{k+1}}(o,\cdot) \leftarrow \check{Q}_{\theta^i_{k,L}}(o,\cdot)$ 
        \STATE $\forall i$ :  $\pi^i_{k+1}(a|o) \gets $ softmax$\left(\frac{1}{\tau_q}\check{Q}_{\theta^i_{k+1}}(o,\cdot)\right)(a)$\label{policy_update_line}
        \STATE $\forall i :$ $\sigma^i_{k+1} \gets 0$\label{start_eval_line}
        \FOR{$e = 0,\dots,E-1$ evaluation steps}
        \STATE Take step $\forall i : a^i_t \sim \pi^i_{k}(\cdot|o^i_t), r^i_{t} = R(s^i_{t},a^i_{t},\hat{\mu}_t),s^i_{t+1} \sim P(\cdot|s^i_{t},a^i_{t},\hat{\mu}_t)$\label{step2}
        \STATE $\forall i :$ $\sigma^i_{k+1} = \sigma^i_{k+1} + \gamma^{e}\cdot r^i_{t}$ 
        \STATE $t \leftarrow t + 1$
        \ENDFOR\label{end_eval_line}
        \FOR{$C_p$ rounds}\label{start_comm}
            \STATE $\forall i :$ Broadcast $\sigma^i_{k+1}, \pi^{i}_{k+1}$\label{brodcast_line}
            \STATE $\forall i : J^i_t \leftarrow i \cup \{j \in \mathcal{N} : (i,j) \in \mathcal{E}^{comm}_{t}$\}\label{receive_line}
            \STATE $\forall i:$ Select $\mathrm{adopted}^i \sim$ Pr$\left(\mathrm{adopted}^i = j\right)$ = $\frac{\exp{(\sigma^j_{k+1}}/\tau^{comm}_k)}{\sum_{x\in J^i_t}\exp{(\sigma^x_{k+1}}/\tau^{comm}_k)}$ $\forall j \in J^i_t$ \label{softmax_adoption_prob}
            \STATE $\forall i : \sigma^i_{k+1} \leftarrow \sigma^{\mathrm{adopted}^i}_{k+1}, \pi^i_{k+1} \leftarrow \pi^{\mathrm{adopted}^i}_{k+1}$\label{adopt_line}
            \STATE Take step $\forall i : a^i_t \sim \pi^i_{k}(\cdot|o^i_t), r^i_{t} = R(s^i_{t},a^i_{t},\hat{\mu}_t),s^i_{t+1} \sim P(\cdot|s^i_{t},a^i_{t},\hat{\mu}_t)$; $t \leftarrow t + 1$\label{step3}
            \ENDFOR\label{end_comm}
    \ENDFOR
    \RETURN policies \{$\pi^i_K$\}$_{i=1}^{N}$
    \end{algorithmic}
\end{algorithm}

\section{Learning and policy improvement}\label{learning_and_policy_improvement}

\subsection{Q-network and update}\label{q_update}

 Lines \ref{initalise_line}-\ref{policy_update_line} of our novel Alg. \ref{networked_algorithm_neural_networks} contain the core Q-function/policy update method. Agent $i$ has a neural network parametrised by $\theta^i_k$ to approximate its Q-function: $\check{Q}_{\theta^i_k}(o,\cdot)$. The agent's policy is given by $\pi_{\theta^i_k}(a|o)= $ softmax$\left(\frac{1}{\tau_{q}}\check{Q}_{\theta^i_k}(o,\cdot)\right)(a)$.  We denote the policy $\pi^i_{k}(a|o)$ for simplicity when appropriate. Each agent maintains a buffer (of size $M$) of collected transitions of the form $\left(o^i_t,a^i_t,r^i_t,o^i_{t+1}\right)$. At each iteration $k$, they empty their buffer (Line \ref{empty_buffer}) before collecting $M$ new transitions (Lines \ref{start_m_loop}-\ref{end_m_loop}); each decentralised agent $i$ then trains its Q-network $\check{Q}_{\theta^i_k}$ via $L$ training updates as follows (Lines \ref{beginlearningwithbuffer}-\ref{endlearningwithbuffer}). For training purposes, $i$ also maintains a target network $\check{Q}_{\theta^{i,'}_{k,l}}$ with the same architecture but parameters $\theta^{i,'}_{k,l}$ copied from $\theta^{i}_{k,l}$ less regularly than $\theta^{i}_{k,l}$ themselves are updated, i.e. only every $\nu$ learning iterations (Line \ref{update_target_line}). At each iteration $l$, the agent samples a random batch $B^i_{k,l}$ of $|B|$ transitions from its buffer (Line \ref{sample_from_buffer}),
and trains its neural network via stochastic gradient descent to minimise the empirical loss (Def. \ref{loss}, Line \ref{minimise_loss_line}). For $cl < 0$, $[\cdot]^0_{cl}$ is a clipping function used to prevent numerical issues if the policy is too close to deterministic, as the log-policy term is otherwise unbounded \citep{NEURIPS2020_2c6a0bae,wu2024populationaware}:

\begin{definition}[Empirical loss for Q-network]\label{loss}
For the target $T$:
\begin{align*}
    T = r_t &+ 
    [\tau_{q}\ln\pi_{\theta^{i,'}_{k,l}}(a_t|o_t)]^0_{cl}  \\ &+ \gamma\sum_{a\in\mathcal{A}}\pi_{\theta^{i,'}_{k,l}}(a|o_{t+1})\left(\check{Q}_{\theta^{i,'}_{k,l}}(o_{t+1},a) - \tau_{q}\ln\pi_{\theta^{i,'}_{k,l}}(a|o_{t+1})\right),
\end{align*}
\[\text{the loss is:} \;\;\;\;\;\;\;\;\;\; \hat{\mathcal{L}}(\theta,\theta') = \frac{1}{|B|}\sum_{transition \in B^i_{k,l}}\left|\check{Q}_{\theta^i_{k,l}}(o_t,a_t) - T\right|^2.\] 
\end{definition}

\subsection{Communication and adoption of parameters}\label{policy_adoption_and_communication}

We use the communication network $\mathcal{G}^{comm}_t$ to share two types of information at different points in Alg \ref{networked_algorithm_neural_networks}. One is used to improve local estimates of the mean field (Sec. \ref{estimation}). The other, described here, is used to privilege the spread of better performing policy updates through the population, allowing faster convergence in this networked case than in the independent and even central-agent cases.

We adapt \citet{benjamin2024networked} for the function-approximation case, where in our work agents broadcast the parameters of the Q-network that defines their policy, rather than the Q-function table. At each iteration $k$, after independently updating their Q-network and policy (Lines \ref{empty_buffer}-\ref{policy_update_line}), agents \textit{estimate} the infinite discounted return (Def. \ref{Mean_field_discounted_return}) of their new policies by collecting rewards for $E$ steps, and assign the finite-step discounted sum to $\sigma^i_{k+1}$ (Lines \ref{start_eval_line}-\ref{end_eval_line}). They then broadcast their Q-network parameters along with $\sigma^i_{k+1}$ (Line \ref{brodcast_line}). Receiving these from neighbours on the network, agents select which set of parameters to adopt by taking a softmax over their own and the received estimate values $\sigma^j_{k+1}$ (Lines \ref{receive_line}-\ref{adopt_line}). They repeat the process for $C_p$ rounds
. This allows decentralised agents to adopt policy parameters estimated to perform better than their own, accelerating learning as shown in Sec. \ref{theoretical_results}.

\section{Mean-field estimation and communication}\label{estimation}

We now give our algorithms for decentralised estimation of the empirical categorical mean-field distribution. We first describe the general version, assuming the more general setting where $\mathcal{G}^{obs}_t$ applies (see discussion in Sec. \ref{preliminaries_networks}). We subsequently detail how the algorithm can be made more efficient in environments where the more abstract visibility graph $\mathcal{G}^{vis}_t$ applies, as in our experimental settings. In both cases, the algorithm runs to generate the observation object when a step is taken in the main Alg. \ref{networked_algorithm_neural_networks}, i.e. to produce $o^i_t = (s^i_t, \tilde{\hat{\mu}}^i_t)$ for the steps $a^i_t \sim \pi^i_{k}(\cdot|o^i_t)$ in Lines \ref{step1}, \ref{step2} and \ref{step3}. {Note that if $\mathcal{G}^{obs}_t$/$\mathcal{G}^{vis}_t$ are \textit{fully} connected, all agents' estimated mean-field observations will be equivalent to the true categorical distribution.} Both versions of the algorithm are subject to implicit assumptions, which we discuss methods for addressing as future work in Appx. \ref{future_work}.

\subsection{Algorithm for the general setting}

\begin{algorithm}[!t]
\caption{Mean-field estimation 
in general settings}
\label{alg:mean_field_estimation_general}
\begin{algorithmic}[1]
\REQUIRE Time-dependent observation graph $\mathcal{G}^{obs}_t$, time-dependent communication graph $\mathcal{G}^{comm}_t$, states $\{s^i_t\}_{i=1}^{N}$, number of communication rounds $C_e$
\STATE $\forall i,s :$ Initialise count vector $\hat{\upsilon}^i_t[s]$ with $\emptyset$\label{initialise_empty_general_line}
\STATE $\forall i : $ $\hat{\upsilon}^i_t[s^j_t]$ $\leftarrow$ $\{ID^j\}_{j \in i \cup \{j'\in\mathcal{N} : (i, j') \in \mathcal{E}^{obs}_t\}}$\label{collect_observations_general_line}
\FOR{$c_e$ in $1,\dots,C_e$}
    \STATE $\forall i :$ Broadcast $\hat{\upsilon}^i_{t,c_e}$\label{broadcast_general_line}
    \STATE $\forall i : J^i_t \leftarrow \{j \in \mathcal{N} : (i,j) \in \mathcal{E}^{comm}_t\}$
    \STATE $\forall i,s$ : $\hat{\upsilon}^i_{t,(c_e+1)}[s] \leftarrow \hat{\upsilon}^i_{t,c_e}[s] \cup \{\hat{\upsilon}^j_{t,c_e}[s]\}_{j \in J^i_t}$ \label{merge_general_line}
\ENDFOR

\STATE $\forall i :$ $counted\_agents^i_t \leftarrow \sum_{s \in \mathcal{S} : \hat{\upsilon}^i_t[s] \neq \emptyset} |\hat{\upsilon}^i_t[s]|$
\STATE $\forall i :$ $uncounted\_agents^i_t \leftarrow N - counted\_agents^i_t$

\STATE $\forall i,s$ : $\tilde{\hat{\mu}}^i_t[s] \leftarrow \frac{uncounted\_agents^i_t}{N \times |\mathcal{S}|}$\label{uniformly_distribute_line}
\STATE $\forall i,s$ where $\hat{\upsilon}^i_t[s]$ is not $\emptyset$ : $\tilde{\hat{\mu}}^i_t[s] \leftarrow \tilde{\hat{\mu}}^i_t[s] + \frac{|\hat{\upsilon}^i_t[s]|}{N}$\label{not_empty_general_line}

\RETURN mean-field estimates  $\{\tilde{\hat{\mu}}^i_t\}_{i=1}^{N}$ 
\end{algorithmic}
\end{algorithm}


\noindent In this setting, our method (Alg. \ref{alg:mean_field_estimation_general}) assumes each agent is associated with a unique ID to avoid the same agents being counted multiple times. Each agent maintains a `count' vector ${\hat{\upsilon}}^i_t$ of length $|\mathcal{S}|$ i.e. of the same shape as the vector denoting the true empirical categorical distribution of agents. Each state position in the vector can hold a list of IDs. Before any actions are taken at each time step $t$, each agent's count vector ${\hat{\upsilon}}^i_t$ is initialised as full of $\emptyset$ (`no count') markers for each state (Line \ref{initialise_empty_general_line}). Then, for each agent $j$ with which agent $i$ is connected via the observation graph, $i$ places $j$'s unique ID in its count vector in the correct state position (Line \ref{collect_observations_general_line}). Next, for $C_e \geq 0$ communication rounds, agents exchange their local counts with neighbours on the communication network (Line \ref{broadcast_general_line}), and merge these counts with their own count vector, filtering out the unique IDs of those that have already been counted (Line \ref{merge_general_line}). If $C_e = 0$ then the local count will remain purely independent. By exchanging these partially filled vectors, agents are able to improve their local counts by adding the states of agents that they have not been able to observe directly themselves. 

After the $C_e$ communication rounds, each state position ${\hat{\upsilon}}^i_t[s]$ either still maintains the $\emptyset$ marker if no agents have been counted in this state, or contains $x_s > 0$ unique IDs. The local mean-field estimate $\tilde{\hat{\mu}}^i_t$ is then obtained from ${\hat{\upsilon}}^i_t$ as follows. All states that have a count $x_s$ have this count converted into the categorical probability $x_s / N$ (we assume that agents know the total number of agents in the finite population, even if they cannot observe them all at each $t$) (Line \ref{not_empty_general_line}). The total number of agents counted in ${\hat{\upsilon}}^i_t$ is given by $counted\_agents$ = $\sum_{s\in\mathcal{S}}x_s$, and the agents that have not been observed are $uncounted\_agents$ = $N$ - $counted\_agents$. In this general setting, the unobserved agents are assumed to be uniformly distributed across all the states, so $uncounted\_agents / (N \times |\mathcal{S}|)$ is added to all the values in $\tilde{\hat{\mu}}^i_t$, replacing the $\emptyset$ marker for states for which no agents have been observed (Line \ref{uniformly_distribute_line}).

\begin{algorithm}[!t]
\caption{Mean-field estimation 
for environments with $\mathcal{G}^{vis}_t$}
\label{alg:mean_field_estimation_specific}
\begin{algorithmic}[1]
\REQUIRE Time-dependent visibility graph $\mathcal{G}^{vis}_t$, time-dependent communication graph $\mathcal{G}^{comm}_t$, states $\{s^i_t\}_{i=1}^{N}$, number of communication rounds $C_e$
\STATE $\forall i,s :$ Initialise count vector $\hat{\upsilon}^i_t[s]$ with $\emptyset$
\STATE $\forall i$ and $\forall s' \in \mathcal{S'} : (s^i_t, s') \in \mathcal{E}^{vis}_t$ : $\hat{\upsilon}^i_t[s']$ $\leftarrow$ $\sum_{j \in \{1,\dots,N\} : s^j_t = s'} 1$
\FOR{$c_e$ in $1,\dots,C_e$}
    \STATE $\forall i :$ Broadcast $\hat{\upsilon}^i_{t,c_e}$
    \STATE $\forall i : J^i_t = i \cup \{j \in \mathcal{N} : (i,j) \in \mathcal{E}^{comm}_t\}$
    \STATE $\forall i,s :$ Initialise new count vector $\hat{\upsilon}^i_{t,(c_e+1)}[s]$ with $\emptyset$
    \STATE $\forall i,s$ and $\forall j \in J^i_t : \hat{\upsilon}^i_{t,(c_e+1)}[s] \leftarrow \hat{\upsilon}^j_{t,c_e}[s]$ if $\hat{\upsilon}^j_{t,c_e}[s] \neq \emptyset$
\ENDFOR
\STATE $\forall i :$ $counted\_agents^i_t \leftarrow \sum_{s \in \mathcal{S} : \hat{\upsilon}^i_t[s] \neq \emptyset} \hat{\upsilon}^i_t[s]$
\STATE $\forall i :$ $uncounted\_agents^i_t \leftarrow N - counted\_agents^i_t$
\STATE $\forall i :$ $unseen\_states^i_t \leftarrow \sum_{s \in \mathcal{S} : \hat{\upsilon}^i_t[s] = \emptyset} 1$
\STATE $\forall i,s$ where $\hat{\upsilon}^i_t[s]$ is not $\emptyset$ : $\tilde{\hat{\mu}}^i_t[s] \leftarrow \frac{\hat{\upsilon}^i_t[s]}{N}$
\STATE $\forall i,s$ where $\hat{\upsilon}^i_t[s]$ is $\emptyset$ : $\tilde{\hat{\mu}}^i_t[s] \leftarrow \frac{uncounted\_agents^i_t}{N \times unseen\_states^i_t}$\label{distribute_where_empty_specific_line}
\RETURN mean-field estimates  $\{\tilde{\hat{\mu}}^i_t\}_{i=1}^{N}$ 
\end{algorithmic}
\end{algorithm}

\subsection{Algorithm for visibility-based environments}

We explain now the differences in our estimation algorithm (Alg. \ref{alg:mean_field_estimation_specific}) for the subclass of environments where $\mathcal{G}^{vis}_t$ applies in place of $\mathcal{G}^{obs}_t$, i.e. the mutual observability of agents depends in turn on the  mutual visibility of states. The benefit of $\mathcal{G}^{vis}_t$ over $\mathcal{G}^{obs}_t$ is that the former allows an agent in state $s$ to obtain a correct, complete count $x_{s'} \geq 0$ of all the agents in state $s'$, for any state $s'$ that is visible to $s$ (note the count may be zero). Unique IDs are thus not required as there is no risk of counting the same agent twice when receiving communicated counts: either \textit{all} agents in $s'$ have been counted, or no count has yet been obtained for $s'$. This simplifies the algorithm and helps preserve agent anonymity and privacy.

Secondly, uncounted agents cannot be in states for which a count has already been obtained, since the count is complete and correct, even if the count is $x_{s'} = 0$. Therefore after the $C_e$ communication rounds, the $uncounted\_agents$ proportion needs to be uniformly distributed only across the positions in the vector that still have the $\emptyset$ marker (Line \ref{distribute_where_empty_specific_line}), and not across all states as in the general setting. This makes the estimation more accurate in this setting.\footnote{In our Algs. \ref{alg:mean_field_estimation_general} and \ref{alg:mean_field_estimation_specific}, agents share their local \textit{counts} with neighbours on the communication network $\mathcal{G}^{comm}_{t}$, and only after the $C_e$ communication rounds do they complete their estimated distribution by distributing the uncounted agents along their vectors. An alternative would be for each agent to immediately form a local \textit{estimate} from their local count obtained via $\mathcal{G}^{obs}_{t}$ or $\mathcal{G}^{vis}_{t}$, which is only then communicated and updated via the communication network. However, we take the former approach to avoid poor local estimations spreading through the network and leading to widespread inaccuracies. Information that is certain (the count) is spread as widely as possible, before being locally converted into an estimate of the total mean field. The same would be the case in our extension proposed in Appx. \ref{future_work} for averaging noisy counts, i.e. only the counts would be averaged, with the estimates completed by distributing the remaining agents after the $C_e$ rounds.}

\section{Theoretical results}\label{theoretical_results}

To demonstrate the benefits of the networked architecture by comparison, we also consider 
the results of baseline central-agent and independent architectures given by alternative versions of our algorithm. As in previous MFG works \citep{policy_mirror_independent, benjamin2024networked}, in the central-agent setting, the Q-network updates of arbitrary agent $i=1$ are automatically pushed to all other agents, and the true global mean-field distribution is always used in place of the local estimate i.e. $\tilde{\hat{\mu}}^i_t$ = ${\hat{\mu}}_t$. In the independent case, there are no links in $\mathcal{G}^{comm}_t$ or $\mathcal{G}^{obs}_t$/$\mathcal{G}^{vis}_t$, i.e. $\mathcal{E}^{comm}_t = \mathcal{E}^{obs}_t = \mathcal{E}^{vis}_t = \emptyset$.

Networked populations often learn faster than central-agent ones in our experiments. To indicate how this is possible while allowing simplicity of the theory, we give a proof for a special case with relatively strong assumptions that give conditions under which networked populations \textit{do} outperform central-agent ones. Nevertheless the intuition provided by Thm. \ref{faster_learning_theorem}'s proof suggests why networked agents can learn faster even without enforcing the assumptions, and we discuss loosening them subsequently.  

Recall that at each iteration $k$ of Alg. \ref{networked_algorithm_neural_networks}, after independently updating their policies in Line \ref{policy_update_line}, the population has the policies \{$\pi^i_{k+1}$\}$_{i=1}^{N}$. There is randomness in these independent policy updates, stemming from the random sampling of each agent's independently collected buffer. In Lines \ref{start_eval_line}-\ref{end_eval_line}, agents estimate the infinite discounted returns \{$V(\pi^i_{k+1},I(\pi^i_{k+1}))\}_{i=1}^{N}$ (Def. \ref{Mean_field_discounted_return}) of their updated policies by computing \{$\sigma^i_{k+1}$\}$_{i=1}^{N}$: the $E$-step discounted return with respect to the \textit{empirical} mean field generated when agents follow policies \{$\pi^i_{k+1}$\}$_{i=1}^{N}$ (i.e. they do not at this stage all follow a single identical policy). We can characterise the approximation as \{$\sigma^i_{k+1}$\}$_{i=1}^{N}$ = \{$\widehat{V}(\pi^i_{k+1},I(\pi^i_{k+1}))\}_{i=1}^{N}$. We now assume the following:

\begin{assumption}\label{single_policy_assumption}
After the $C_p$ rounds in Lines \ref{start_comm}-\ref{end_comm} in which agents exchange and adopt policies from neighbours, the population is left with a single policy such that $\forall i,j \in \{1,\dots,N\}$ $\pi^i_{k+1} = \pi^j_{k+1}$.\footnote{Most simply Assumption \ref{single_policy_assumption} holds if 1) $\tau^{comm}_{k}$ is a positive constant sufficiently close to zero that the softmax essentially becomes a max function, and 2) the communication network $\mathcal{G}^{comm}_t$ is static and connected during the $C_p$ communication rounds, where $C_p$ is larger than the network diameter $d_{\mathcal{G}^{comm}_t}$. Under these conditions, previous results on max-consensus algorithms show that all agents in the network will converge on the highest $\sigma^{max}_{k+1}$ value (and hence the unique associated $\pi^{max}_{k+1}$) within a number of rounds equal to the diameter $d_{\mathcal{G}^{comm}_t}$ \citep{5348437,benjamin2024networked}. However, policy consensus 
can be achieved even outside of these conditions, including if the network is dynamic and not connected at every step, given appropriate values for $C_p$ and $\tau^{comm}_{k+1} \in \mathbb{R}_{>0}$.}
\end{assumption}

\begin{assumption}\label{approximation_ordering_assumption}
Assume that \{$\sigma^i_{k+1}$\}$_{i=1}^{N}$ are sufficiently good approximations so as to respect the ordering of the true values\\ \{$V(\pi^i_{k+1},I(\pi^i_{k+1}))\}_{i=1}^{N}$, i.e. $\forall i,j \in \{1,\dots,N\}$:
\[\sigma^i_{k+1} > \sigma^j_{k+1} \iff V(\pi^i_{k+1},I(\pi^i_{k+1})) > V(\pi^j_{k+1},I(\pi^j_{k+1})).\]
\end{assumption}

Call the network consensus policy $\pi^{\mathrm{net}}_{k+1}$, and its associated finitely estimated return $\sigma^{\mathrm{net}}_{k+1}$. Recall that the central-agent case is where the Q-network update of arbitrary agent $i=1$ is automatically pushed to all the others instead of the policy exchange in Lines \ref{start_eval_line}-\ref{end_comm}; this is equivalent to a networked case where policy consensus is reached on a \textit{random} one of the policies \{$\pi^i_{k+1}$\}$_{i=1}^{N}$. Call this policy \textit{arbitrarily} given to the whole population $\pi^{\mathrm{cent}}_{k+1}$, and its associated finitely estimated return $\sigma^{\mathrm{cent}}_{k+1}$. Now we can say:

\begin{theorem}\label{faster_learning_theorem}
    Given Assumptions \ref{single_policy_assumption} and \ref{approximation_ordering_assumption}, \[\mathbb{E}[V(\pi^{\mathrm{net}}_{k+1},I(\pi^{\mathrm{net}}_{k+1}))] > \mathbb{E}[V(\pi^{\mathrm{cent}}_{k+1},I(\pi^{\mathrm{cent}}_{k+1}))].\] Thus in expectation networked populations will increase their returns faster than central-agent ones.
\end{theorem}

\begin{proof}\label{proof_of_theorem} Recall that before the communication rounds in Line \ref{start_comm} (Alg. \ref{networked_algorithm_neural_networks}), the  randomly updated policies \{$\pi^i_{k+1}$\}$_{i=1}^{N}$ have  associated estimated returns \{$\sigma^i_{k+1}$\}$_{i=1}^{N}$. Call the mean and maximum of this set $\sigma^{\mathrm{mean}}_{k+1}$ and $\sigma^{\mathrm{max}}_{k+1}$ respectively. Since $\pi^{\mathrm{cent}}_{k+1}$ is chosen arbitrarily from \{$\pi^i_{k+1}$\}$_{i=1}^{N}$, it will obey $\mathbb{E}[\sigma^{\mathrm{cent}}_{k+1}] = \sigma^{\mathrm{mean}}_{k+1}$ $\forall k$, though there will be high variance. Conversely, the softmax adoption probability (Line \ref{softmax_adoption_prob}, Alg. \ref{networked_algorithm_neural_networks}) for the networked case means by definition that policies with higher $\sigma^i_{k+1}$ are more likely to be adopted at each communication round. Thus the $\pi^{\mathrm{net}}_{k+1}$ that is adopted by the whole networked population will obey $\mathbb{E}[\sigma^{\mathrm{net}}_{k+1}] > \sigma^{\mathrm{mean}}_{k+1}$ (if $\tau^{comm}_{k+1}$ is a positive value near zero, it will obey $\mathbb{E}[\sigma^{\mathrm{net}}_{k+1}] = \sigma^{\mathrm{max}}_{k+1}$ $\forall k$). So $\mathbb{E}[\sigma^{\mathrm{net}}_{k+1}] > \mathbb{E}[\sigma^{\mathrm{cent}}_{k+1}]$, which by Assumption \ref{approximation_ordering_assumption} implies the result.
\end{proof}

The adoption scheme in Line \ref{softmax_adoption_prob} biases the spread of policies towards those estimated to be better, which, given sufficiently good approximations (Assumption \ref{approximation_ordering_assumption}), results in higher discounted returns in practice. By choosing updates in a more principled way, networked agents learn faster than the central-agent case that pushes updates regardless of quality.\footnote{Assumption \ref{single_policy_assumption} may not hold if $C_p$ is not large enough or if parts of the population remain isolated. Thus in our experiments, where we use $C_p = 1$ to show the benefit of even one communication round, networked populations with smaller broadcast radii outperform central-agent populations by a smaller margin. Nevertheless the intuition provided by  Thm. \ref{faster_learning_theorem}’s proof indicates how the former are still able to perform better even \textbf{if Assumption \ref{single_policy_assumption} is loosened}.}\footnote{Even \textbf{if Assumption \ref{approximation_ordering_assumption} does not strictly hold}, the softmax parameter $\tau^{comm}_{k}$ allows a smooth degradation as the ordering of the approximations worsens with respect to the ordering of the true values. I.e. if instead of the exact correct policy ordering we have that better policies are simply \textit{more likely} to be given higher estimated evaluations, then the softmax means that these policies remain \textit{more likely} to spread, and the best policy may still be adopted even if it is not evaluated as the best.} 
Similar logic can be applied to understand why networked agents outperform independent ones, coupled with the fact that greater policy diversity in the independent case worsens sample complexity over the networked and central-agent cases by biasing approximations of the Q-function \citep{policy_mirror_independent,benjamin2024networked}.

Significantly, our communication scheme not only avoids the undesirable assumption of a central node, but even outperforms it. Moreover, the benefit of the scheme over central-agent learning is greater with our function approximation than in the tabular case (cf. \citep{benjamin2024networked}), perhaps due to greater variance in the quality of Q-function estimates in our case. This shows that networked communication facilitates greater scalability than the central-agent paradigm.

\section{Experiments}\label{experiments_section}

We provide two sets of experiments. The first shows that our function-approximation algorithm (Alg. \ref{networked_algorithm_neural_networks}) can scale to large state spaces for population-independent policies, and that in such settings networked, communicating populations can outperform purely-independent and even central-agent populations
. The second set demonstrates that Alg. \ref{networked_algorithm_neural_networks} can handle population-dependent policies, as well as the ability of Alg.  \ref{alg:mean_field_estimation_specific} to 
estimate the mean field locally. 

\subsection{Experimental set-up}

For the types of game in our experiments we follow the gold standard in prior MFG works, i.e. grid-world environments where agents can move in the four cardinal directions or remain in place \citep{lauriere2021numerical,scalable_deep,zaman2023oraclefree,algumaei2023regularization,cui2023multiagent,benjamin2024networked,wu2024populationaware}. We present results from four tasks defined by the agents' reward/transition functions, all of which are \textit{coordination} tasks - see Appx. \ref{objectives_appendix} for full technical descriptions, and also a fifth, non-coordination task. In all cases, rewards are normalised in [0,1] after they are computed. The first two tasks are those used with population-independent policies in \citet{benjamin2024networked}, but while they show results for an 8x8 and a `larger' 16x16 grid, our results are for $100\times 100$ and $50\times 50$ grids:

\begin{itemize}
    \item \textbf{Cluster.} Agents are rewarded for gathering but given no indication where to do so, agreeing it over time.

\item \textbf{Target agreement.} Agents are rewarded for visiting any of a given number of targets, but the reward is proportional to the number of other agents co-located at the target. Agents must coordinate on which single target they will all meet at to maximise their individual rewards.
\end{itemize}

We also show our algorithms handling more complex tasks
:

\begin{itemize}
    \item \textbf{Evade shark in shoal.} At each $t$, a ‘shark’ takes a step towards the grid point containing the most agents according to the empirical mean field. The shark's position is part of agents' local states in addition to their own position. Agents are rewarded more for being further from the shark and for clustering with other agents. As well as having a non-stationary distribution, we add ‘common noise’, with the shark taking a random step with probability 0.01. Such noise that affects the local states of all agents in the same way, making the evolution of the distribution stochastic, makes population-independent policies sub-optimal \citep{survey_learningMFGs}.

\item \textbf{Push object to edge.} Agents are rewarded for 
how close an ‘object’ is to the grid's edge. The object's position forms part of agents' local states in addition to their own position. The object moves in a direction with a probability proportional to the number of agents on its opposite side, i.e. agents must coordinate on which side of the object from which to ‘push’ it, to ensure it moves toward the edge of the grid.
\end{itemize}

We evaluate our experiments via two metrics. 
\textit{Exploitability} is the most common metric in MFG works, and is a measure of proximity to the MFNE. It quantifies how much an agent can benefit by deviating from the set of policies that generate the current mean field, with a lower exploitability meaning the population is closer to the MFNE. However, there are several issues with this metric in our setting, particularly for our coordination tasks where competitive agents benefit from aligning behaviours, such that it may give limited or noisy information (discussed further in Appx. \ref{exploitability}). We thus also give a second metric, as in \citet{benjamin2024networked}: the population’s \textit{average discounted return}. {This allows us to compare how quickly agents are learning to increase their returns, even when exploitability gives us limited ability to distinguish between the desirability of the MFNEs to which populations converge.}

\subsection{Results and discussion}\label{results}

In our spatial environments, the physical distance from $i$ determines the communication graph $\mathcal{G}^{comm}_t$ and the visibility graph $\mathcal{G}^{vis}_t$. Our plots show various radii, given as fractions of the maximum possible distance (the grid's diagonal length). We set $C_p=C_e=1$ to show the benefit to learning speed brought by even a \textit{single} communication round. Note that the networked population with the largest radius is always fully connected, and therefore these agents are always able to accurately estimate ${\hat{\mu}}_t$ even for $C_e=0$. That is, their observations are equivalent to those that the central-agent population would receive, albeit that policies are updated and spread differently. Other hyperparameters are detailed in Appx. \ref{hyperparameters_section}.

\begin{figure}[t]
\centering
\includegraphics[width=0.98\columnwidth]{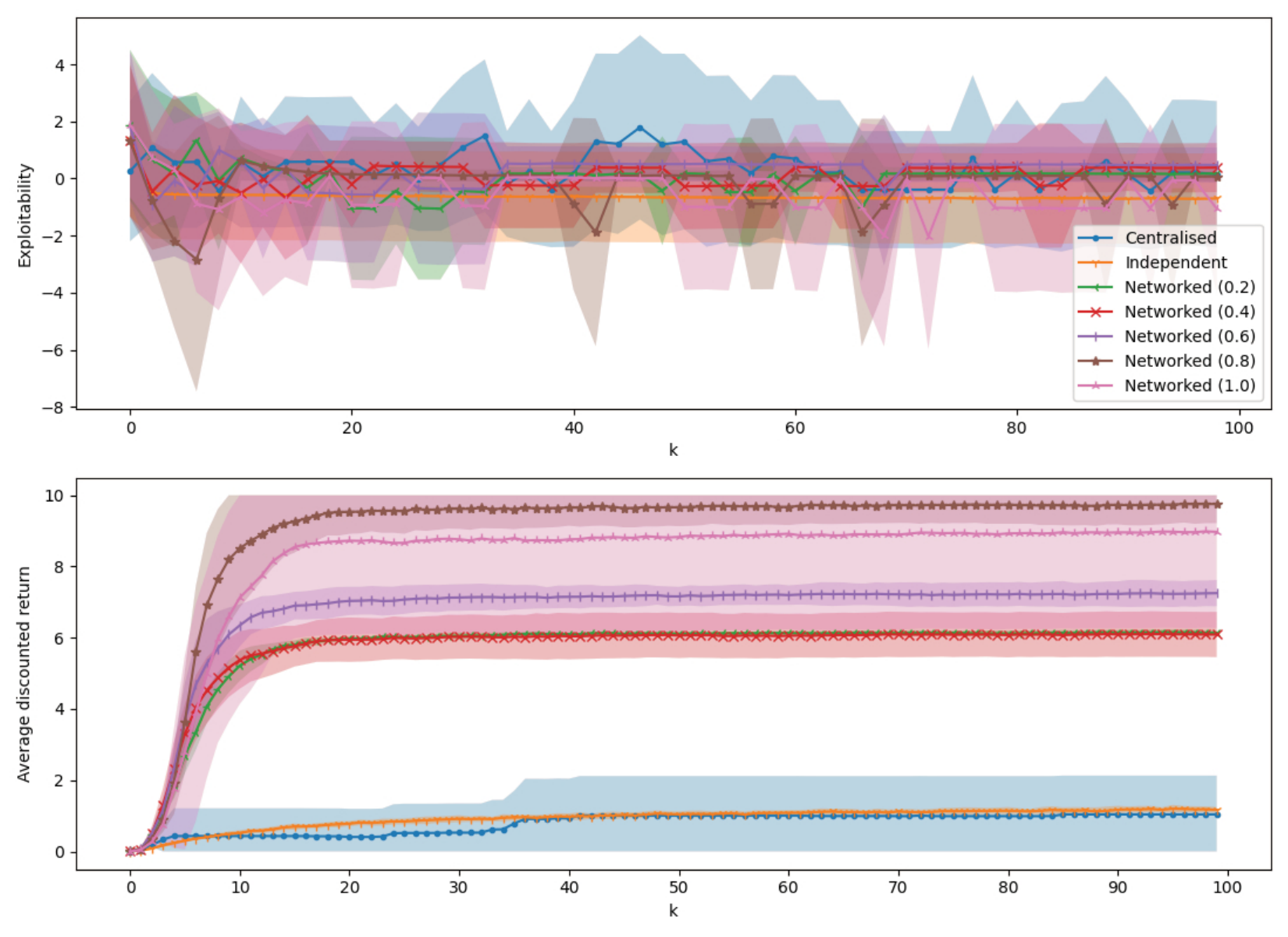}
\caption{`Target agreement', pop.-independent, $100\times 100$ grid. Reproduced larger in Fig. \ref{agree100bigger}. The networked populations of all broadcast radii significantly outperform the central-agent and independent populations in terms of average return, where the latter two cases hardly appear to learn at all.}
\label{agree100}
\Description{`Target agreement', population-independent, $100\times 100$ grid. The networked populations of all broadcast radii significantly outperform the central-agent and independent populations in terms of average return, where the latter two cases hardly appear to learn at all.}
\end{figure}

\begin{figure}[t]
\centering
\includegraphics[width=0.98\columnwidth]{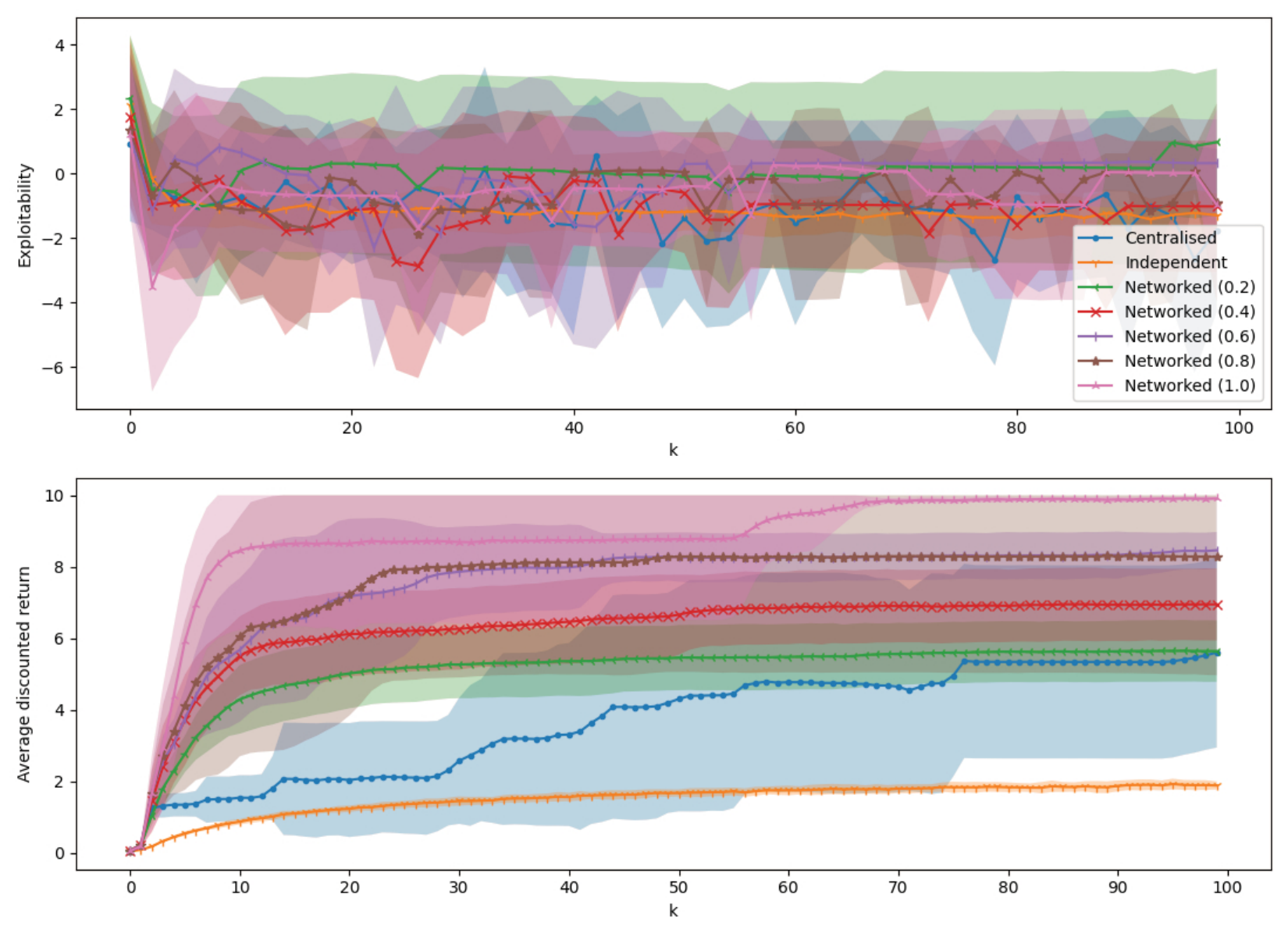}
\caption{`Cluster', pop.-independent, $100\times 100$ grid. Larger version in Fig. \ref{cluster100bigger}. Networked populations of all broadcast radii outperform the central-agent and independent populations wrt. average return; independent agents hardly learn at all.}
\label{cluster100}\Description{`Cluster', population-independent, $100\times 100$ grid. The networked populations of all broadcast radii outperform the central-agent and independent populations in terms of average return; independent agents hardly appear to learn at all.}
\end{figure}

\subsubsection{Population-independent policies in large state-spaces}

Figs. \ref{agree100} and \ref{cluster100} (for $100\times 100$ and grids; reproduced larger in Figs. \ref{agree100bigger} and \ref{cluster100bigger}, Appx. \ref{experiments_appendix}), and Figs. \ref{agree50} and \ref{cluster50} (for $50\times 50$ grids; Appx. \ref{experiments_appendix})  
illustrate that introducing function approximation to algorithms in this setting allows them to converge within a practical number of iterations ($k\ll100$), even for large state spaces. By contrast, the tabular algorithms in \citet{benjamin2024networked} appear only just to converge by $k=200$ for the same tasks for the {larger} of their two grids, which is only $16\times16$. 

In Figs. \ref{agree100}, \ref{cluster100} and \ref{agree50}, no populations appear to have significantly different exploitability to each other, while in Fig. \ref{cluster50} the central-agent population may have lower exploitability, but not significantly so. As discussed in Appx. \ref{exploitability}, the exploitability metric is noisy and provides limited information in these coordination games. Nevertheless we can see that in all four plots the independent agents hardly improve their returns at all, while the central-agent populations hardly improve their returns in the `target agreement' games in Figs. \ref{agree100} and \ref{agree50}. There is therefore little a deviating agent can do to increase its return in these coordination games, meaning exploitability appears low, despite these being undesirable equilibria. 

Meanwhile, the networked agents do improve their returns and thus significantly outperform the stagnant independent agents in Figs. \ref{agree100}, \ref{cluster100}, \ref{agree50} and \ref{cluster50} and the stagnant central-agent populations in Figs. \ref{agree100} and \ref{agree50}. This indicates that our communication scheme helps agents to reach ‘preferable’ equilibria, even when exploitability is similar (for reasons related to how networked populations can increase their returns faster as per Sec. \ref{theoretical_results}). While central-agent populations do appear to increase their returns in the `cluster' task in Figs. \ref{cluster100} and \ref{cluster50}, they do so more slowly and reaching a lower final value than all networked agents in the $100\times100$ grid case, and than all networked agents apart from those with the smallest broadcast radii in the $50\times50$ grid case. Indeed in the $100\times100$ grids in Figs. \ref{agree100} and \ref{cluster100}, the central-agent populations appear to perform less well than they do in the $50\times50$ grids in Figs. \ref{agree50} and \ref{cluster50}, whereas the networked populations do not suffer a performance decrease, indicating that our networked communication scheme scales better and is robust to larger state spaces than the central-agent paradigm. Similarly, in the `target agreement' and `cluster' tasks in the tabular setting in \citet{benjamin2024networked}, the central-agent populations generally perform similarly to the networked populations, indicating that our networked architecture is more robust than the central-agent alternative when moving to non-tabular settings. 


\subsubsection{Population-dependent policies in complex environments}

We also show the ability of our algorithms to handle more complex tasks, using population-dependent policies and estimated mean-field observations. Figs. \ref{push,estimate} and \ref{evade,estimate} (Appx. \ref{experiments_appendix}), where agents estimate the mean field via Alg. \ref{alg:mean_field_estimation_specific}, differ minimally from Figs. \ref{push,global} and \ref{evade,global}, where agents directly receive the global mean field. This indicates that our estimation algorithm allows agents to appropriately estimate the distribution, even with only one round of communication for agents to help each other improve their local counts. Only in the `push object' task in Fig. \ref{push,estimate}, and there only with the smaller broadcast radii, do agents slightly underperform the returns of agents in the global observability case in Fig. \ref{push,global}, as might be expected.

For the reasons given in Appx. \ref{exploitability} regarding coordination games, the exploitability metric gives limited information in the `push object' and `evade' tasks in Figs. \ref{push_task} and  \ref{evade_task}: for example, the return of a best-responding agent in the `push object' task still depends on the extent to which other agents coordinate on which direction in which to push the box, meaning it cannot significantly increase its return by deviating. However, all of the networked cases significantly outperform the independent learners in terms of the average return to which they converge in both tasks. In the `push object' task networked learners also appear to outperform central-agent populations, while in the `evade' task all networked cases perform similarly to the central-agent case. Recall though that in the real world a central-agent architecture is a strong assumption, a computational bottleneck and single point of failure.

\section{Conclusion}\label{conclusion}

We introduced function approximation to the online-learning setting for empirical MFGs, and also contributed two novel algorithms for locally estimating the empirical mean field for populat\-ion-dependent policies. We proved theoretically that our networked communication algorithms can learn faster than even central-agent 
architectures in this 
setting, and showed empirically that they can handle large state spaces and estimate the mean field. For future work, see Appx. \ref{future_work}.

\bibliographystyle{ACM-Reference-Format} 
\bibliography{main}


\newpage\;\newpage
\appendix
\section*{\huge Appendices}

\section{Experimental set-up}\label{experiments_appendix}

Experiments were conducted on a Linux-based machine with 2 x Intel Xeon Gold 6248 CPUs (40 physical cores, 80 threads total, 55 MiB L3 cache). We use the JAX framework to accelerate and vectorise our code.  For reproducibility, our code is included in our Supplementary Material. Random seeds are set in our code in a fixed way dependent on the trial number to allow easy replication of experiments. 

\begin{figure*}[t]
\centering
\includegraphics[width=\textwidth
]{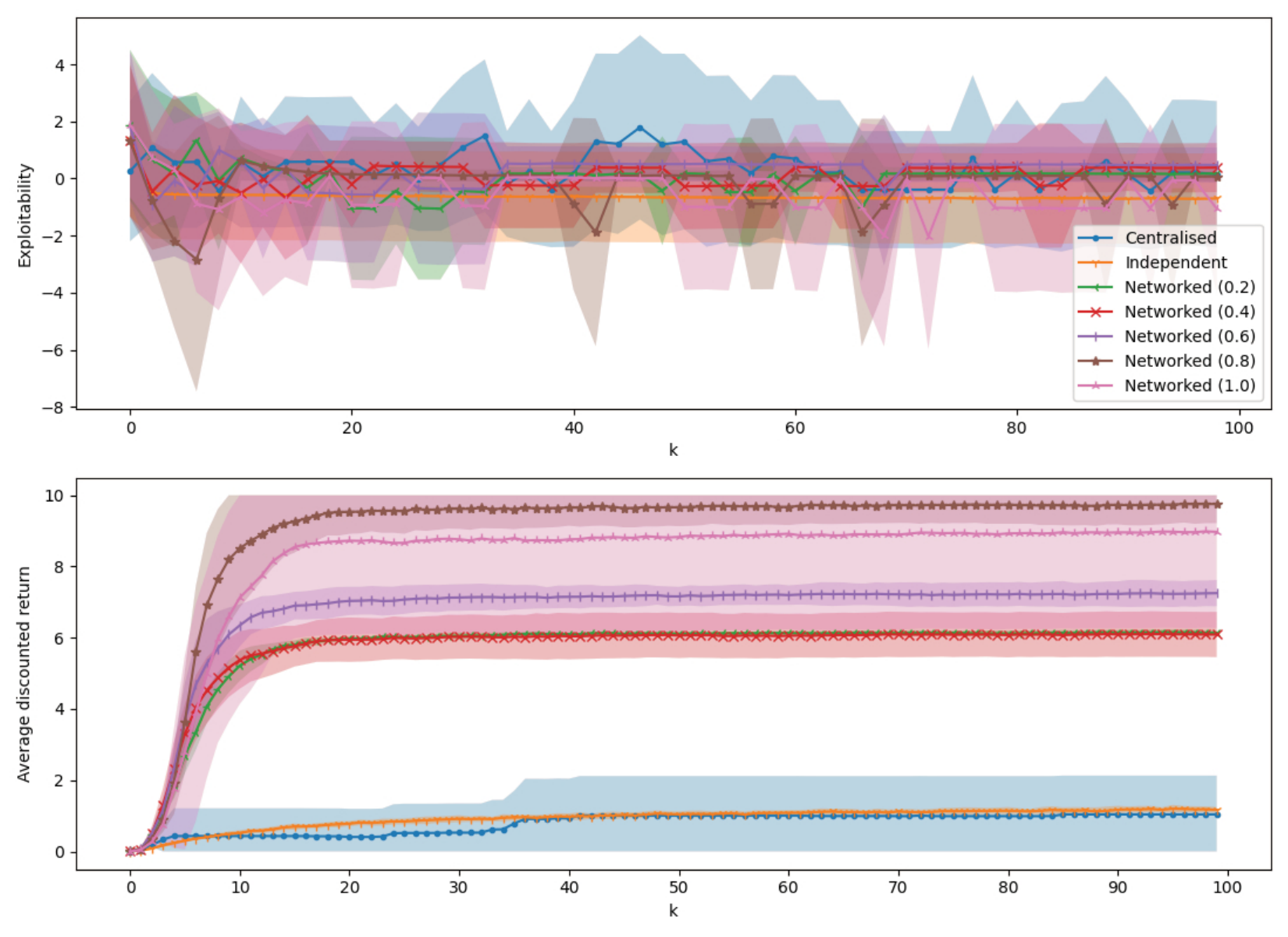}
\caption{Larger version of Fig. \ref{agree100}. `Target agreement', population-independent, $100\times 100$ grid. The networked populations of all broadcast radii significantly outperform the central-agent and independent populations in terms of average return, where the latter two cases hardly appear to learn at all.}
\label{agree100bigger}
\Description{`Target agreement', population-independent, $100\times 100$ grid. The networked populations of all broadcast radii significantly outperform the central-agent and independent populations in terms of average return, where the latter two cases hardly appear to learn at all.}
\end{figure*}

\begin{figure*}[t]
\centering
\includegraphics[width=\textwidth
]{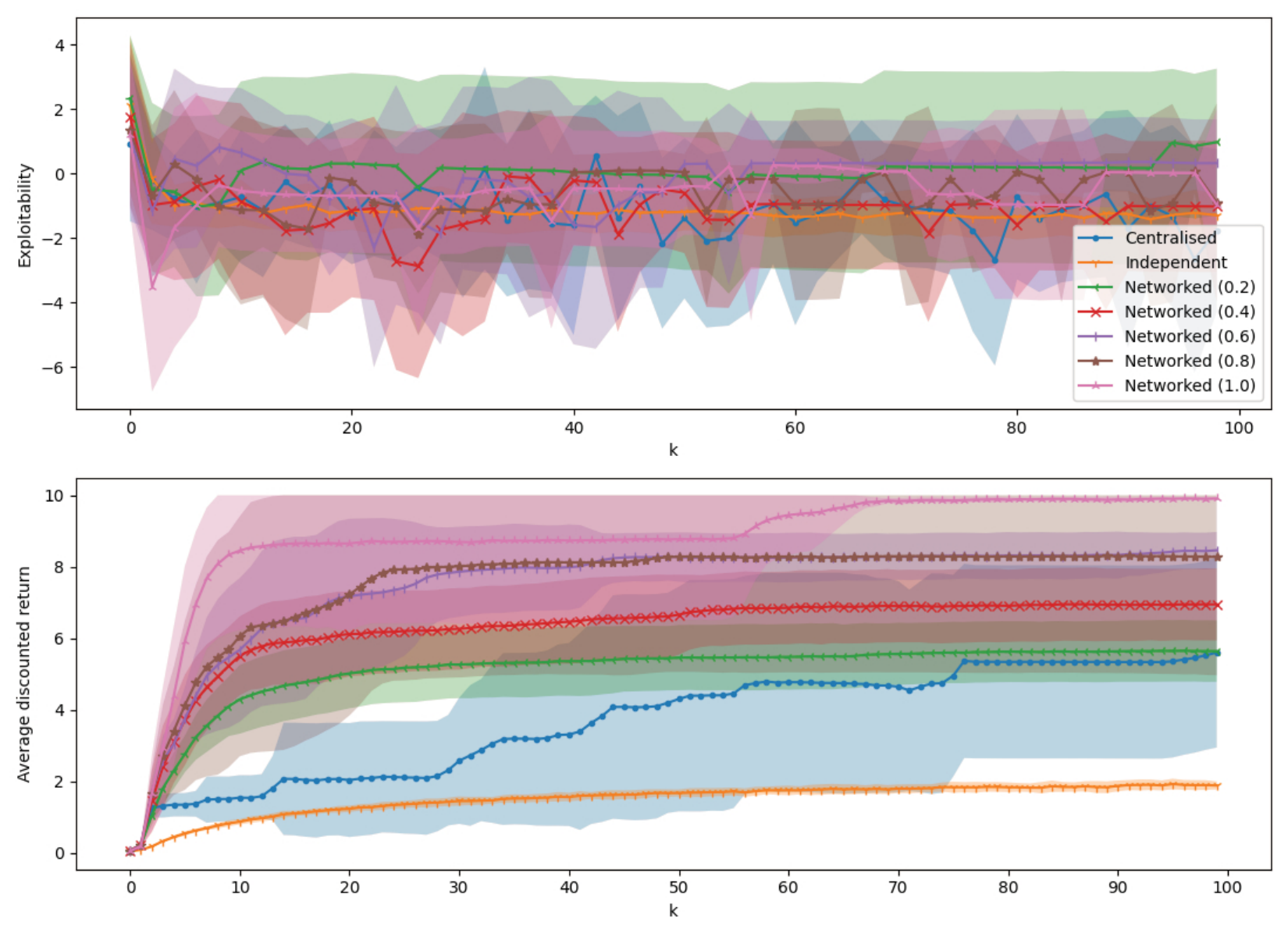}
\caption{Larger version of Fig. \ref{cluster100}. `Cluster', population-independent, $100\times 100$ grid. Networked populations of all broadcast radii outperform the central-agent and independent populations wrt. average return; independent agents hardly appear to learn at all.}
\label{cluster100bigger}\Description{`Cluster', population-independent, $100\times 100$ grid. The networked populations of all broadcast radii outperform the central-agent and independent populations in terms of average return; independent agents hardly appear to learn at all.}
\end{figure*}

\begin{figure*}[t]
\centering
\includegraphics[width=\textwidth
]{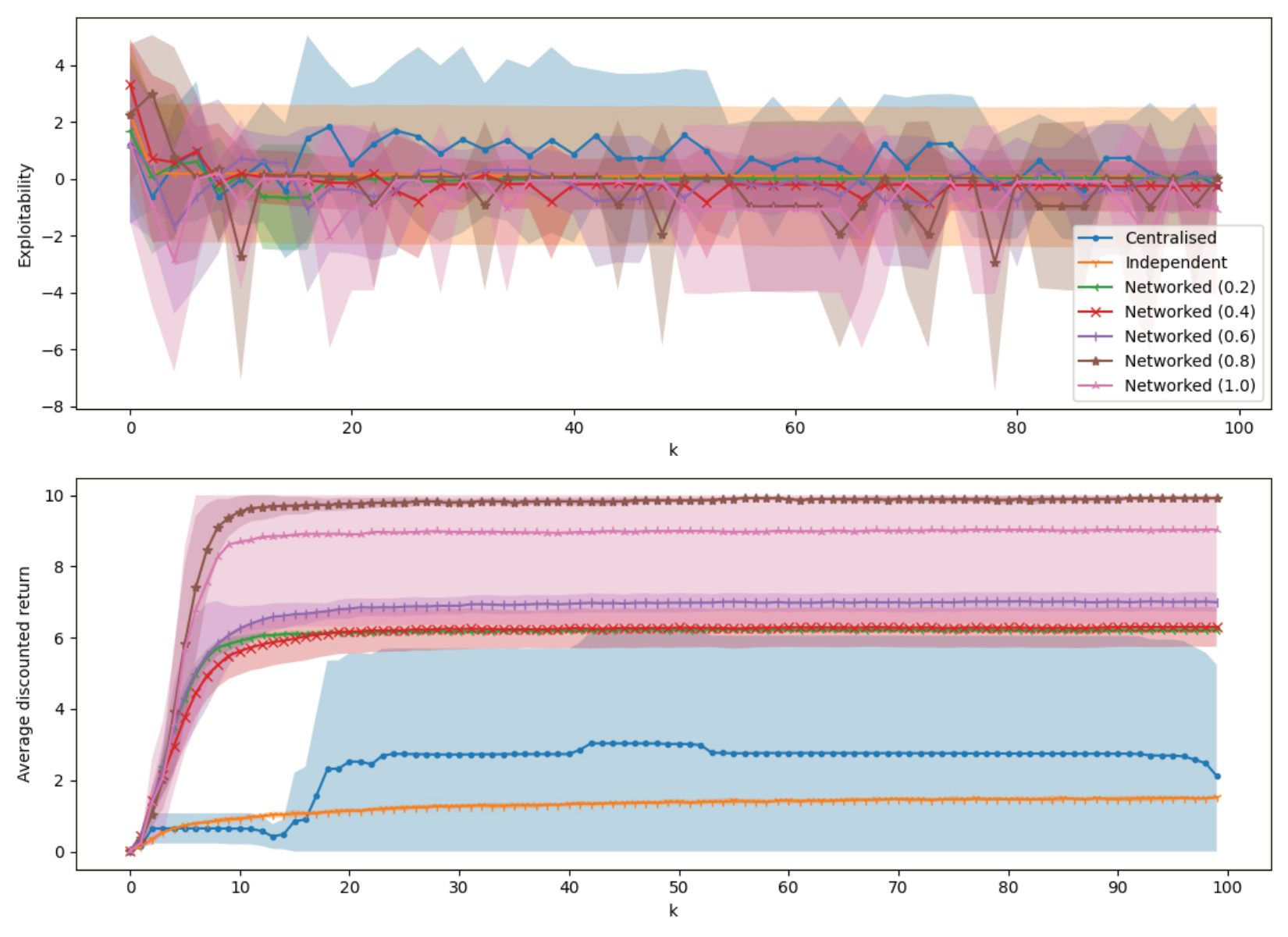}
\caption{`Target agreement' task, population-independent policies, $50\times 50$ grid. The networked populations of all broadcast radii significantly outperform the central-agent and independent populations in terms of average return, where the latter two cases hardly appear to learn at all.}
\label{agree50}\Description{`Target agreement' task, population-independent policies, $50\times 50$ grid. The networked populations of all broadcast radii significantly outperform the central-agent and independent populations in terms of average return, where the latter two cases hardly appear to learn at all.}
\end{figure*}

\begin{figure*}[t]
\centering
\includegraphics[width=\textwidth
]{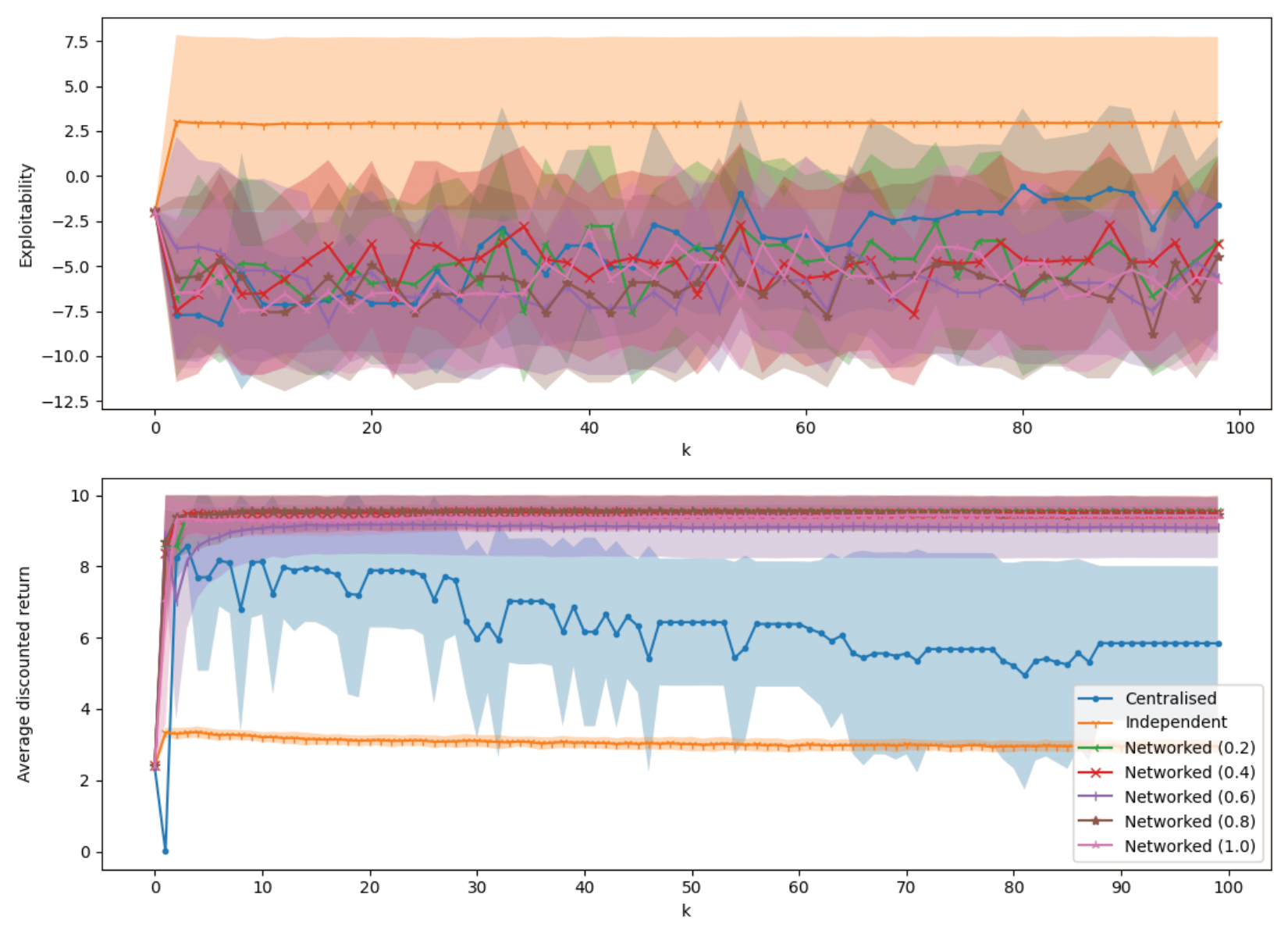}
\caption{`Disperse' task, population-independent policies, $100\times 100$ grid. The networked populations of all broadcast radii significantly outperform the central-agent and independent populations in terms of average return (and exploitability in the case of independent agents), with independent agents not learning at all.}
\label{diffuse100}\Description{`Disperse' task, population-independent policies, $100\times 100$ grid. The networked populations of all broadcast radii significantly outperform the central-agent and independent populations in terms of average return (and exploitability in the case of independent agents), with independent agents not learning at all.}
\end{figure*}

\begin{figure*}[t]
\centering
\includegraphics[width=\textwidth
]{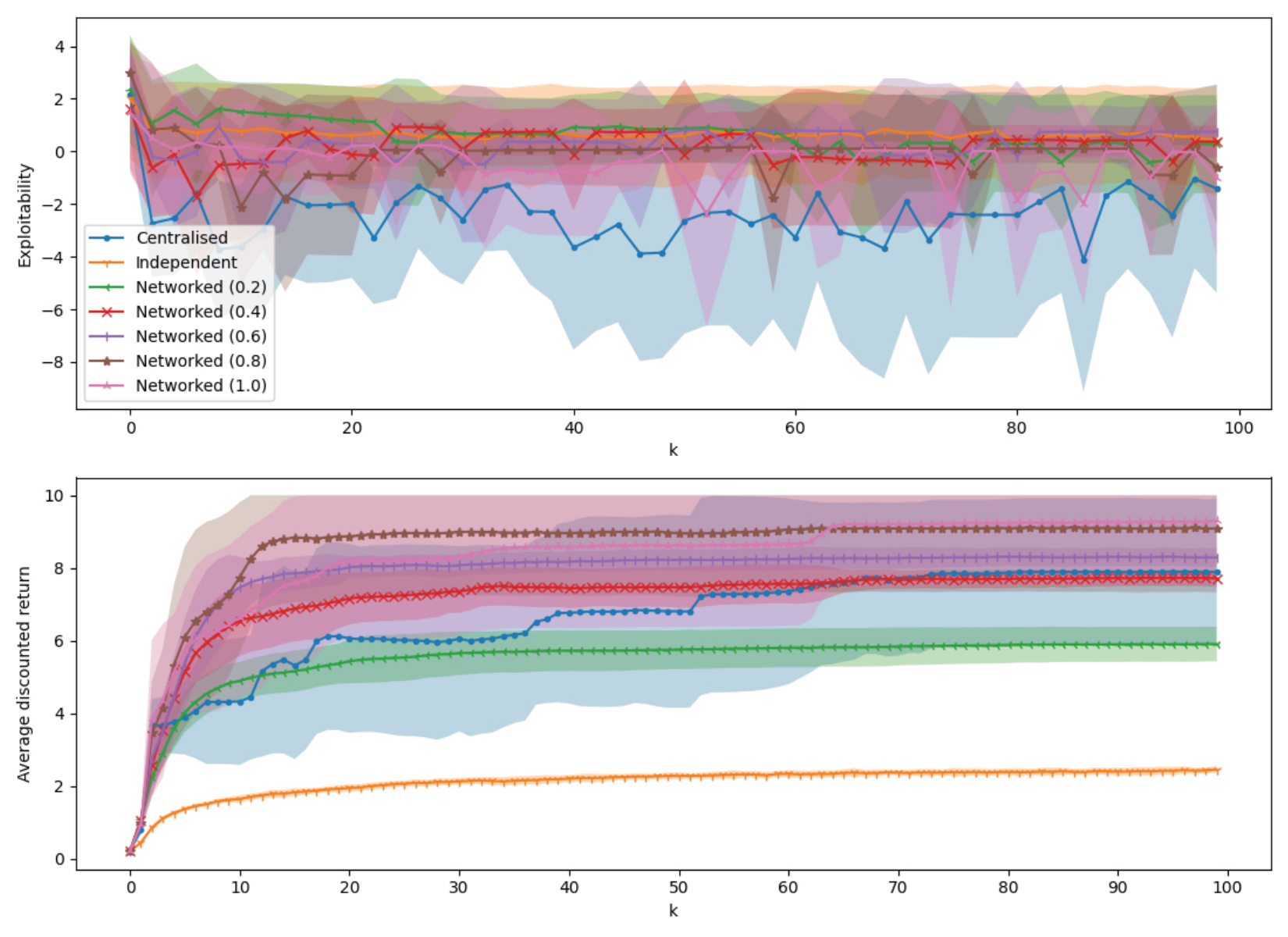}
\caption{`Cluster' task, population-independent policies, $50\times 50$ grid. The networked agents of all broadcast radii significantly outperform the independent agents in terms of average return, where the latter case hardly appears to learn at all. The higher broadcast radii also appear to outperform the central-agent case in terms of average return, with the latter outperforming all others in terms of exploitability.}
\label{cluster50}\Description{`Cluster' task, population-independent policies, $50\times 50$ grid. The networked agents of all broadcast radii significantly outperform the independent agents in terms of average return, where the latter case hardly appears to learn at all. The higher broadcast radii also appear to outperform the central-agent case in terms of average return, with the latter outperforming all others in terms of exploitability.}
\end{figure*}

\begin{figure*}[t]
    \centering
    \begin{subfigure}[b]{0.49\textwidth
    }
        \centering
        \includegraphics[width=\textwidth]{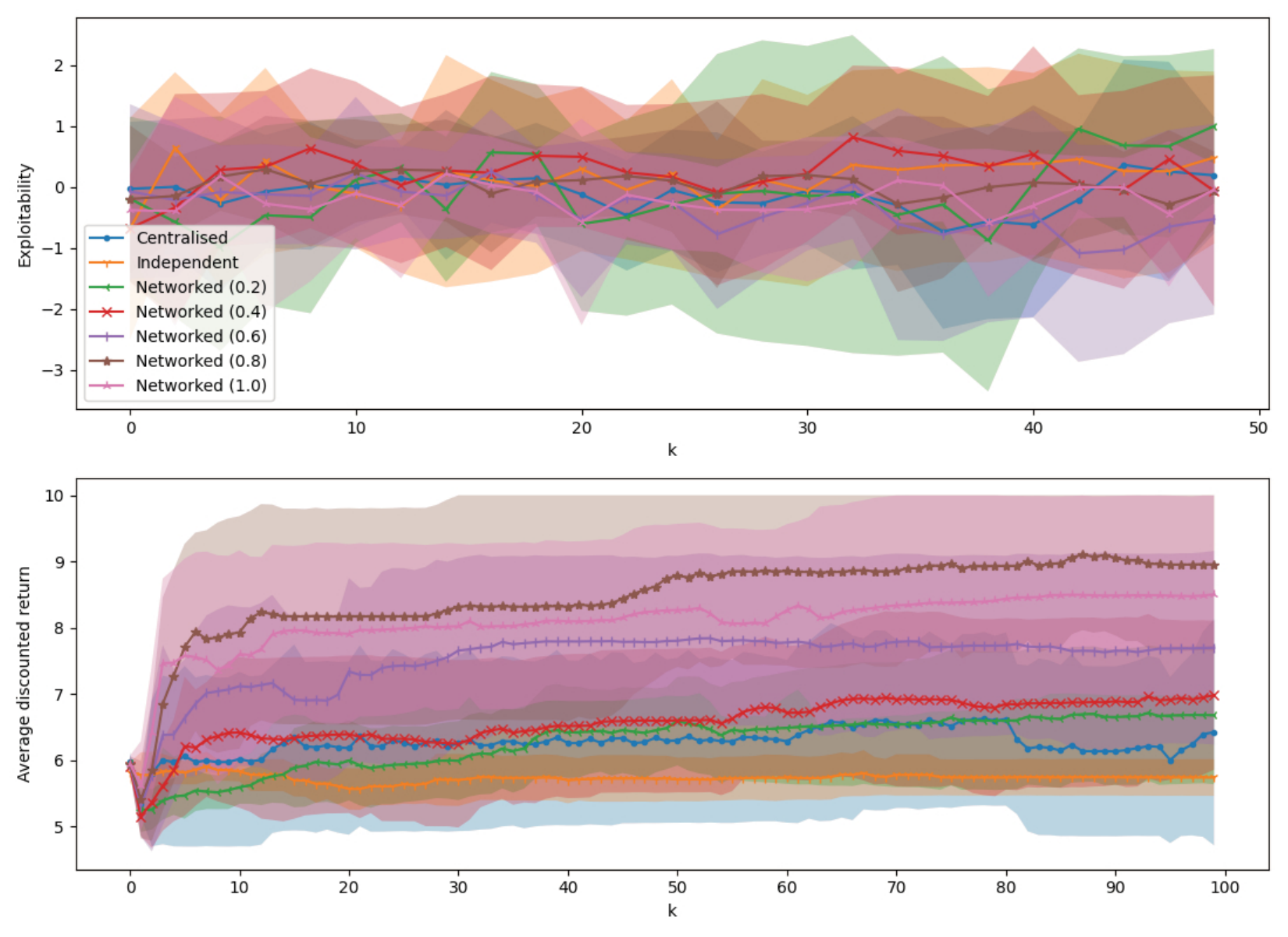}
        \caption{Estimated mean-field distribution.}
        \label{push,estimate}
    \end{subfigure}
    \hfill
    \begin{subfigure}[b]{0.49\textwidth
    }
        \centering
        \includegraphics[width=\textwidth]{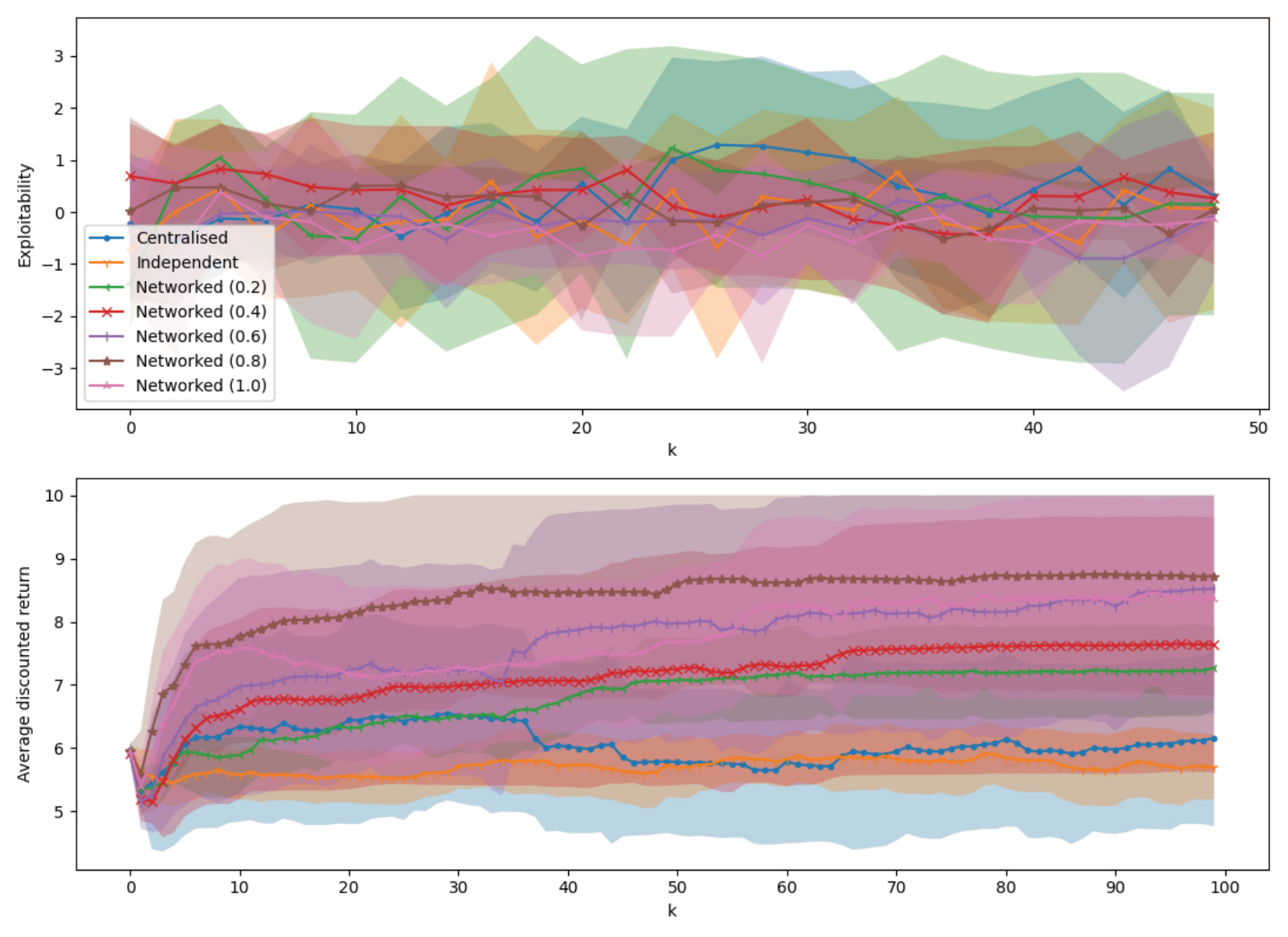}
        \caption{Global observability of mean field.}
        \label{push,global}
    \end{subfigure}
    \caption{`Push object' task, population-dependent policies on a $10\times 10$ grid. The networked populations of all broadcast radii appear to outperform the central-agent and independent populations in terms of average return; the latter two cases hardly appear to learn at all.}
    \label{push_task}\Description{`Push object' task, population-dependent policies on a $10\times 10$ grid. The networked populations of all broadcast radii appear to outperform the central-agent and independent populations in terms of average return; the latter two cases hardly appear to learn at all.}
\end{figure*}

\begin{figure*}[t]
    \centering
    \begin{subfigure}[b]{0.49\textwidth
    }
        \centering
        \includegraphics[width=\textwidth]{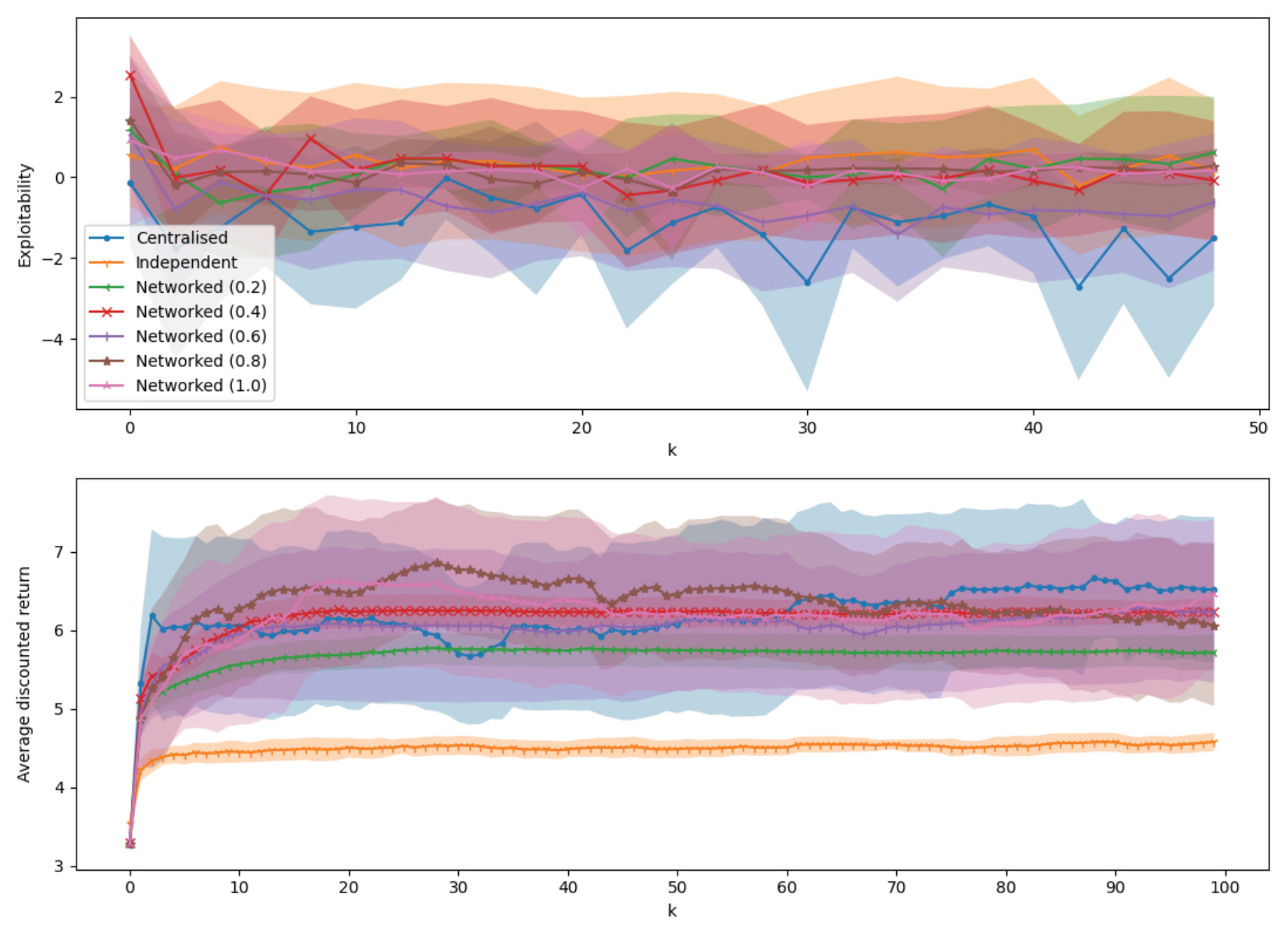}
        \caption{Estimated mean-field distribution.}
        \label{evade,estimate}
    \end{subfigure}
    \hfill
    \begin{subfigure}[b]{0.49\textwidth
    }
        \centering
        \includegraphics[width=\textwidth]{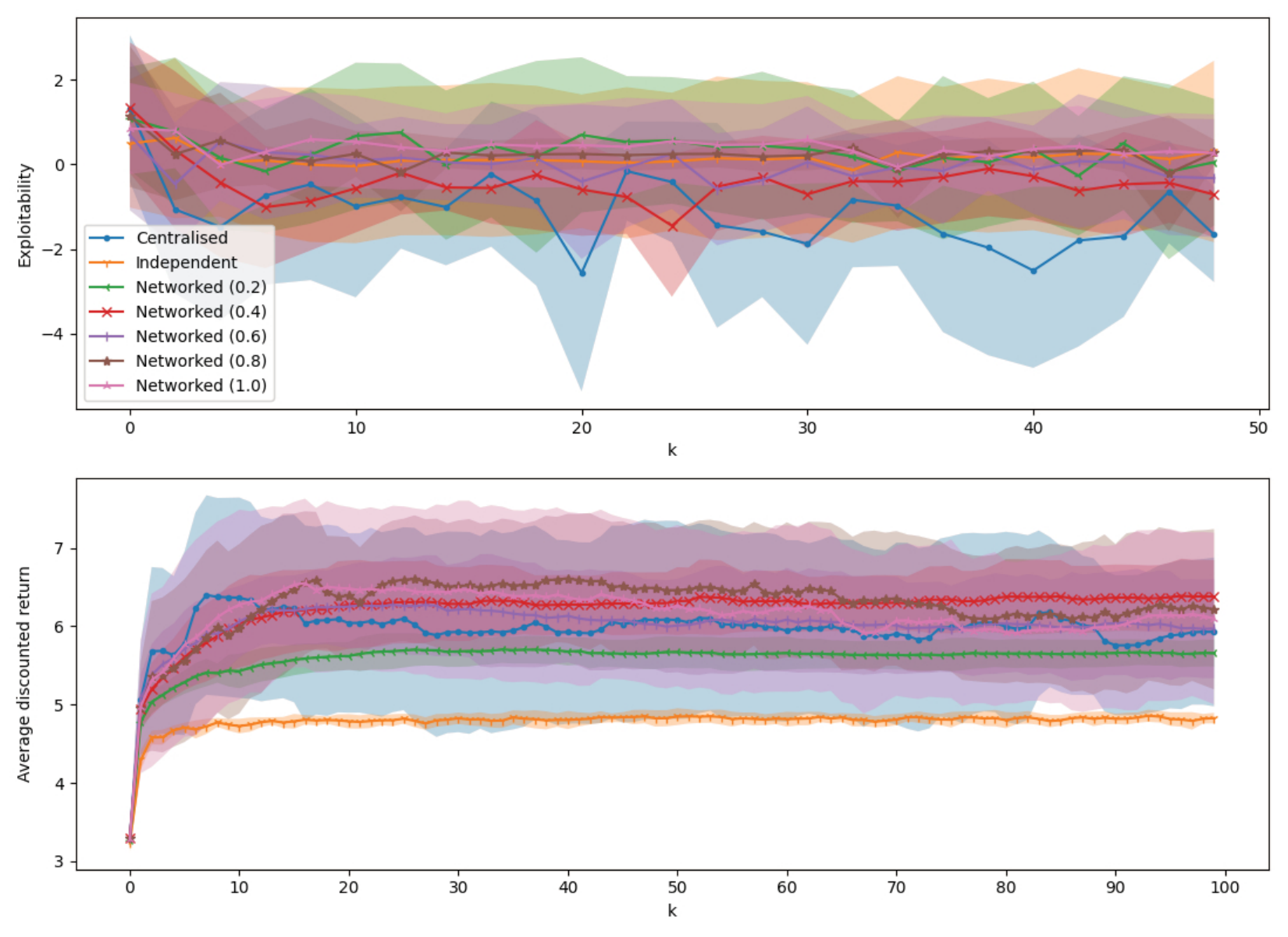}
        \caption{Global observability of mean field.}
        \label{evade,global}
    \end{subfigure}
    \caption{`Evade' task, population-dependent policies on a $10\times 10$ grid. The networked agents of all broadcast radii significantly outperform the independent agents in terms of average return, and perform similarly to the central-agent populations.}
    \label{evade_task}\Description{`Evade' task, population-dependent policies on a $10\times 10$ grid. The networked agents of all broadcast radii significantly outperform the independent agents in terms of average return, and perform similarly to the central-agent populations.}
\end{figure*}

\subsection{Technical description of tasks}\label{objectives_appendix}

\paragraph{Cluster.} This is the inverse of the `exploration' game in \citet{scalable_deep}, where in our case agents are encouraged to gather together by the reward function $R(s^i_{t},a^i_{t},\hat{\mu}_t) =$ log$(\hat{\mu}_t(s^i_{t}))$. That is, agent $i$ receives a reward that is logarithmically proportional to the fraction of the population that is co-located with it at time $t$. We give the population no indication where they should cluster, agreeing this themselves over time.

\paragraph{Target agreement.} Unlike in the above `cluster' game, the agents are given options of locations at which to gather, and they must reach consensus among themselves. If the agents are co-located with one of a number of specified targets $\phi \in \Phi$ (in our experiments we place one target in each of the four corners of the grid), and other agents are also at that target, they get a reward proportional to the fraction of the population found there; otherwise they receive a penalty of -1. In other words, the agents must coordinate on which of a number of mutually beneficial points will be their single gathering place to maximise their individual rewards. Define the magnitude of the distances between $x,y$ at $t$ as $dist_t(x,y)$. The reward function is given by $R(s^i_{t},a^i_{t},\hat{\mu}_t) = r_{targ}(r_{collab}(\hat{\mu}_t(s^i_{t})))$, where \[r_{targ}(x)=\begin{cases}
          x \quad &\text{if} \, \exists \phi \in \Phi \text{ s.t. dist}_t(s^i_{t},\phi) = 0 \\
          -1 \quad &  \text{otherwise,} \\
     \end{cases} \]
\[r_{collab}(x)=\begin{cases}
          x \quad &\text{if} \, \hat{\mu}_t(s^i_{t}) > 1/N \\
          -1 \quad &  \text{otherwise.} \\
     \end{cases}\]

 \paragraph{Evade shark in shoal.}  This is a similar idea to the task found in \citet{liu2025achievement}. Define the magnitude of the horizontal and vertical distances between $x,y$ at $t$ as $dist^h_t(x,y)$ and $dist^v_t(x,y)$ respectively. The state $s^i_{t}$ now consists of agent $i$'s position $x^i_{t}$ and a `shark's' position $\phi_{t}$. At each time step, the shark steps towards the most populated grid point according to the empirical mean-field distribution i.e. $x^*_t = \arg\max_{x\in\mathcal{S}} \hat{\mu}_t(x)$. A horizontal step is taken if $dist^h_t(\phi_{t},x^*_t)$ $\geq$ $dist^v_t(\phi_{t},x^*_t)$, otherwise a vertical step is taken. As well as featuring a non-stationary distribution, we add ‘common noise’ to the environment, with the shark moving in a random direction with probability 0.01. Such noise that affects the local states of all agents in the same way, making the evolution of the distribution stochastic, makes population-independent policies sub-optimal \citep{survey_learningMFGs}. Agents are rewarded more for being further from the shark, and also for clustering with other agents. The reward function is given by 
\begin{align*}
    R(s^i_{t},a^i_{t},\hat{\mu}_t) = \;& dist^h_t(\phi_{t},x^i_t) + \\ &dist^v_t(\phi_{t},x^i_t) + \text{norm}_{dist}(\log(\hat{\mu}_t(x^i_{t}))),
\end{align*}
where $\text{norm}_{dist}(\cdot)$ indicates that the final term is normalised to have the same maximum and minimum values as the total combined vertical and horizontal distance.
    
\paragraph{Push object to edge.}  This is similar to the task presented in \citet{cunha2024occlusion}. As before, define the magnitude of the horizontal and vertical distances between $x,y$ at $t$ as $dist^h_t(x,y)$ and $dist^v_t(x,y)$ respectively. The state $s^i_{t}$ consists of agent $i$'s position $x^i_{t}$ and the object's position $\phi_{t}$. The number of agents in the positions surrounding the object at time $t$ generates a probability field around the object, such that the object is most likely to move in the direction away from the side with the most agents. As such, if agents are equally distributed around the object, it will be equally likely to move in any direction, but if they coordinate on choosing the same side, they can `push' it in a certain direction. If Edges = \{edge$^1,\dots,$edge$^4\}$ are the grid edges, the closest edge to the object at time $t$ is given by edge$^*_t = \arg\min_{\text{edge}\in\text{Edges}}\left(\min(dist^h_t(\phi_{t},\text{edge}),dist^h_t(\phi_{t},\text{edge})\right).$ Agents are rewarded for how close they are to the object, and for how close the object is to the edge of the grid, i.e. they must coordinate on which side of the object from which to ‘push’ it, to ensure it moves to the grid's edge. The reward function is given by
\begin{align*}
    R(s^i_{t},a^i_{t},\hat{\mu}_t) = \;& dist^h_t(\phi_{t},x^i_t) + dist^v_t(\phi_{t},x^i_t) + \\ & dist^h_t(\phi_{t},\text{edge}^*_t) + dist^v_t(\phi_{t},\text{edge}^*_t).
\end{align*}

The above tasks are all coordination tasks, in that agents receive higher returns by aligning their policies and hence have an incentive to communicate their policy parameters even though the MFG framework is non-cooperative. Our fifth game is a non-coordination task: the reward function is not designed to give higher returns for more aligned policies. We include this game to demonstrate that even in such non-cooperative scenarios our networked architecture receives higher returns than both independent and central-agent alternatives, such that decentralised selfish agents may still have incentive to communicate. The fifth game is:

\paragraph{Disperse.} This is similar to the `exploration' tasks in \citet{scalable_deep}, \citet{wu2024populationaware} and other MFG works. In our version agents are rewarded for being located in more sparsely populated areas, but only if they are stationary. The reward function is given by $R(s^i_{t},a^i_{t},\hat{\mu}_t) = r_{stationary}(-\hat{\mu}_t(s^i_{t}))$, where \[r_{stationary}(x)=\begin{cases}
          x \quad &\text{if} \, a^i_{t} \text{ is `remain stationary'}  \\
          -1 \quad &  \text{otherwise.} \\
     \end{cases} \]

\subsection{Experimental metrics}\label{metrics_appendix}

To give as informative results as possible about both performance and proximity to the MFNE, we provide two metrics for each experiment. Both metrics are plotted with mean and standard deviation, computed over the ten trials (each with a random seed) of the system run in each setting.

\subsubsection{Exploitability}\label{exploitability}Works on MFGs most commonly use the \textit{exploitability} metric to evaluate how close a given policy $\pi$ is to a NE policy $\pi^*$ \citep{scalable_deep,continuous_fictitious_play,survey_learningMFGs,algumaei2023regularization,scaling_MFG,benjamin2024networked,wu2024populationaware}. The metric usually assumes that all agents are following the same policy $\pi$, and quantifies how much an agent can benefit by deviating from $\pi$ by measuring the difference between the return given by $\pi$ and that of a $BR$ policy with respect to the distribution generated by $\pi$: 

\begin{definition}[Exploitability of $\pi$]\label{exploitability_definition}
The exploitability ${E}x$ of policy $\pi$ is given by:
\[{E}x(\pi) = V(BR(I(\pi)),I(\pi)) - V(\pi,I(\pi)).\] 
\end{definition}

If $\pi$ has a large exploitability then an agent can significantly improve its return by deviating from $\pi$, meaning that $\pi$ is far from $\pi^*$, whereas an exploitability of 0 implies that $\pi = \pi^*$. Prior works conducting empirical testing have generally focused on the central-agent setting, so this classical definition, as well as most evaluations, only consider exploitability when all agents are following a single policy $\pi_k$. However, \citet{benjamin2024networked} notes that purely independent agents, as well as networked agents, may have divergent policies $\pi^i_k \neq \pi^j_k$ $\forall i,k \in 1,\dots,N$, as in our own setting. We therefore are interested in the `exploitability' of the population's joint policy $\boldsymbol\pi$ := ($\pi^1,\dots,\pi^N$) $\in \Pi^N$. 

Since we do not have access to the exact $BR$ policy as in some related works \citep{scalable_deep,wu2024populationaware}, we must instead approximate the exploitability, similarly to \citet{flock,benjamin2024networked}. We freeze the policy of all agents apart from a deviating agent, for which we store its current policy and then conduct 50 $k$ loops of policy improvement. To approximate the expectations in Def. \ref{exploitability_definition}, we take the best return of the deviating agent across 10 additional $k$ loops, as well as the mean of all the other agents' returns across these same 10 loops. (While the policies of all non-deviating agents is $\pi_k$ in the central-agent case, if the non-deviating agents do not share a single policy, then this method is in fact approximating the exploitability of their joint policy $\boldsymbol\pi^{-d}_k$, where $d$ is the deviating agent.) We then revert the agent back to its stored policy, before learning continues for all agents as per the main algorithm. Due to the expensive computations required for this metric, we evaluate it every second $k$ iteration of the main algorithm for Figs. \ref{agree100}, \ref{cluster100}, \ref{agree50},  \ref{cluster50}   and \ref{diffuse100}, and every fourth iteration for the population-dependent experiments. 

\textit{The exploitability metric has a number of limitations in our setting}. In coordination games (the setting for all tasks apart from the `disperse' game), agents benefit by following the same behaviour as others, and so a deviating agent generally stands to gain less from switching to a $BR$ policy than it might in the non-coordination games on which many other works focus. For example, the return of a best-responding agent in the `cluster’ task still depends on the extent to which other agents coordinate on where to cluster, meaning it cannot significantly increase its return by deviating from a badly clustering policy. This means that the downward trajectory of the exploitability metric is less clear in our plots than in other works that do not focus on coordination games. This is likely why the difference between the approximated exploitability of the independent agents and the other populations is clearer in the non-coordination `disperse' task in Fig. \ref{diffuse100} than in the other tasks.

Moreover, our approximation takes place via MOMD policy improvement steps (as in the main algorithm) for an independent, deviating agent while the policies of the rest of the population are frozen. As such, the quality of our approximation is limited by the number of policy-improvement/expectation-estimation rounds, which must be restricted for the sake of the running speed of the experiments. Moreover, since one of the findings of our paper is that independent-learning agents increase their returns significantly slower (if at all) than networked or central-agent populations, it is arguably unsurprising that approximating the $BR$ by an independently deviating agent sometimes gives an unclear and noisy metric. This includes the exploitability going below zero, which should not be possible if the policies and distributions are computed exactly. Given the limitations presented by approximating exploitability, we also provide a second metric to indicate the progress of learning, as in \citet{benjamin2024networked}.

\subsubsection{Average discounted return}

We record the average discounted return of the agents' policies $\pi_k^i$ during the $M$ iterations. This allows us to compare how quickly agents are learning to increase their returns, even when exploitability gives us limited ability to distinguish between the desirability of the MFNEs to which populations converge. We can observe that settings that converge to similar exploitability values may not have similar average agent returns, suggesting that some algorithms are better than others at finding not just NE, but preferable NE (for reasons related to why networked populations can increase their returns faster as per Sec. \ref{theoretical_results}). See for example Figs. \ref{agree100} and \ref{agree50}, where the networked populations converge to similar exploitability as the independent and central-agent populations, but receive higher average returns.

\subsection{Hyperparameters}\label{hyperparameters_section}

See Table \ref{Hyperparameters} for our hyperparameter choices. We can group our hyperparameters into those controlling the size of the experiment, those controlling the size of the Q-network, those controlling the number of iterations of each loop in the algorithms and those affecting the learning/policy updates or policy adoption.

In our experiments we generally want to demonstrate that our communication-based algorithms outperform the central-agent and independent architectures by allowing policies that are estimated to be better performing to proliferate through the population, such that convergence occurs within fewer iterations and computationally faster, even when the Q-function is poorly approximated and/or the mean-field is poorly estimated, as is likely to be the case in real-world scenarios. Moreover we want to show that there is a benefit even to a small amount of communication, so that communication rounds themselves do not excessively add to time complexity. As such, we generally select hyperparameters at the lowest end of those we tested during development, to show that our algorithms are particularly successful given what might otherwise be considered `undesirable' hyperparameter choices. 

\begin{table*}
    \centering
    \caption{Hyperparameters}\label{Hyperparameters}
    \hspace*{-0.cm}
    \begin{tabular}{p{1.85cm}|p{1.5cm}|p{13.cm}}\\ 
  Hyperparam.     & Value     & Comment \\ \midrule
  Trials & 10  & We run 10 trials with different random seeds for each experiment. {We plot the mean and standard deviation for each metric across the trials.}   \\ \hline
Gridsize     & $10$$\times$$10$/ $50$$\times$$50$/ $100$$\times$$100$ & Experiments on large state spaces are run on $50\times 50$ and $100\times 100$ grids (Figs. \ref{agree100}-\ref{diffuse100}), while experiments with population-dependent policies are run on the $10\times 10$ grid (Figs. \ref{push,estimate}, \ref{evade,estimate},  \ref{push,global} and \ref{evade,global}). \\ \hline

    Population     & 500 & We chose 500 for our demonstrations to show that our algorithm can handle large populations, indeed often larger than those demonstrated in other mean-field works, especially for grid-world environments, while also being feasible to simulate wrt. time and computation constraints \citep{subjective_equilibria,approximately_entropy,cui2023multiagent,general_framework,rl_stationary,mfrl_yang,pomfrl,subramanian2022multi,DMFG,wu2024populationaware,benjamin2024networked}. \\ \hline
    Number of neurons in input layer & cf. comment & The agent's position is represented by two concatenated one-hot vectors indicating the agent's row and column. An additional two such vectors are added for the shark's/object's position in the `evade' and 'push object' tasks. For population-dependent policies, the mean-field distribution is a flattened vector of the same size as the grid. As such, the input size in the `evade' and 'push object' tasks is $[(4 \times \mathrm{dimension}) + (\mathrm{dimension}^2)]$; in the other settings it is $[2 \times \mathrm{dimension}]$.
    \\ \hline
    Neurons per hidden layer & cf. comment & We draw inspiration from common rules of thumb when selecting the number of neurons in hidden layers, e.g. it should be between the number of input neurons and output neurons / it should be 2/3 the size of the input layer plus the size of the output layer / it should be a power of 2 for computational efficiency. Using these rules of thumb as rough heuristics, we select the number of neurons per hidden layer by rounding the size of the input layer down to the nearest power of 2. The layers are all fully connected.
    \\ \hline
    Hidden layers & 2 & We experimented with 2 and 3 hidden layers in the Q-networks. While 3 hidden layers gave similar or slighly better performance, we selected 2 for increased computational speed for conducting our experiments.
     \\ \hline
     Activation function &   ReLU     & This is a common choice in deep RL.

     \\ \hline
    $K$     & 100       & $K$ is chosen to be large enough to see at least one of the metrics converging.\\ \hline 
    $M$     & 50       &  We tested $M$ in \{50,100\} and found that the lower value was sufficient to achieve convergence while minimising training time. It may be possible to converge with even smaller choices of $M$. 
    \\ \hline 
    
     $L$     &   50    &  We tested $L$ in \{50,100\} and found that the lower value was sufficient to achieve convergence while minimising training time. It may be possible to converge with even smaller choices of $L$.
     \\ \hline 
    $E$     &   20    & We tested $E$ in \{20,50,100\}, and choose the lowest value to show the benefit to convergence even from very few evaluation steps. It may be possible to reduce this value further and still achieve similar results. 
     \\\hline $C_p$     &   1   & As in \citet{benjamin2024networked}, we choose this value to show the convergence benefits brought by even a single communication round, even in networks that may have limited connectivity; higher choices are likely to have even better performance. 
    \\\hline 
    $C_e$     &   1   & Similar to $C_p$, we choose this value to show the ability of our algorithm to appropriately estimate the mean field even with only a single communication round, even in networks that may have limited connectivity.
    \\\hline $\gamma$     & 0.9 & Standard choice across RL literature.     \\\hline 
    $\tau_q$     &   0.03   & We tested $\tau_q$ in \{0.01,0.02,0.03,0.04,0.05\}, as well as linearly decreasing $\tau_q$ from 0.05 $\rightarrow$ 0 as $k$ increases. However, only 0.03 gave stable increase in return. Note that this is the value also chosen in \citet{NEURIPS2020_2c6a0bae}. \\\hline 

    $|B|$   &  32    & This is a common choice of batch size that trades off noisy updates and computational efficiency. \\\hline 

    $cl$     &   -1   & We use the same value as in \citet{NEURIPS2020_2c6a0bae}.\\\hline 

    $\nu$     &  $L-1$    & We tested $\nu$ in $\{1,4,20,L-1\}$. We found that in our setting, updating $\theta'\leftarrow\theta$ once per $k$ iteration s.t. $\theta'_{k+1,l}=\theta_{k,l}$ $\forall l$ gave sufficient learning that was similar to the other potential choices of $\nu$, so we do this for simplicity, rather than arbitrarily choosing a frequency to update $\theta'$ during each $k$ loop. Setting the target to be the policy from the previous iteration is similar to the method in \citet{scalable_deep}. Whilst \citet{wu2024populationaware} updates the target within the $L$ loops for stability, we do not find this to be a problem in our experiments. \\\hline

    Optimiser     &   Adam   & As in \citet{NEURIPS2020_2c6a0bae}, we use the Adam optimiser with initial learning rate 0.01. \\\hline

    $\tau^{comm}_k$    &   cf. comment    &  $\tau^{comm}_k$ increases linearly from 0.001 to 1 across the $K$ iterations. This is a simplification of the annealing scheme used in \citet{benjamin2024networked}. Further optimising the annealing process may lead to better results.
    \\ \hline \hline 
    \end{tabular}    
    \end{table*}

\subsection{Additional experiments and plots}\label{additional_experiments}

In the non-coordination `disperse' task in Fig. \ref{diffuse100}, networked agents significantly outperform independent and central-agent populations in terms of average return. They also significantly outperform independent agents in terms of exploitability, and perhaps also central-agent populations though not significantly so. The fact that this happens in this non-coordination, competitive game shows that agents may have an incentive to communicate with each other even if they are self-interested. Indeed the agents learning independently do not appear to improve their returns at all, despite this being the paradigm that might be expected to perform best in a non-cooperative setting.

\section{Further discussion on the necessity of population-dependent policies}\label{Further_discussion_non_stationary_equilibria}

There are numerous reasons why agents may benefit from being able to respond to the current distribution. These include when:

\begin{itemize}
    \item There may in fact be any initial distribution \citep{perrin2022generalization,wu2024populationaware,survey_learningMFGs}.
    \item The solved MFG has a non-stationary equilibrium, e.g. if there is a finite time horizon meaning that actions become more urgent as time passes such that policies (and therefore also the distribution) are non-stationary \citep{scalable_deep,wu2024populationaware}.
    \item The distribution is subject to `common noise'. This is noise that affects the local states of all agents in the same way, making the evolution of the distribution stochastic. In such cases, even if the agent knows the policy used by all other agents, it cannot perfectly predict the evolution of the mean-field distribution, making population-independent policies suboptimal \citep{carmona_common,cardaliaguet2015masterequationconvergenceproblem,survey_learningMFGs,li2024incomplete,becherer2024common}.
    \item The distribution deviates from the equilibrium for some other reason, such that agents need to be able to generalise their response to other (possibly previously unseen) distributions \citep{survey_learningMFGs}. 
\end{itemize}

\section{Future work}\label{future_work}

Our work follows the gold standard in MFGs by presenting experiments on grid world toy environments, albeit we show our algorithms are able to handle much larger and more complex games than prior work. Nevertheless future work lies in moving from these environments to real-world settings. In Sec. \ref{theoretical_results} we give theoretical results showing that our networked algorithm can outperform a central-agent alternative. We leave more general analysis, such as proof of convergence and sample guarantees in the function approximation setting, for future work.

Since the MFG setting is technically non-cooperative, we have preempted objections that agents would not have incentive to communicate their policies by focusing on coordination games, i.e. where agents seek to maximise only their individual returns, but receive higher rewards when they follow the same strategy as other agents. In this case they stand to benefit by exchanging their policies with others. Future work lies in extending our networked communication algorithms to mean-field control, the cooperative counterpart to MFGs, where agents would have incentive to communicate across different types of game. Nevertheless, in real-world settings, the communication network could still be vulnerable to malfunctioning agents or adversarial actors poisoning the equilibrium by broadcasting untrue policy information \citep{agrawal2024impact}, or equally to unreliable communication channels. It is outside the scope of this paper to analyse how much false information would have to be broadcast by how many agents to affect the equilibrium, but real-world applications may need to compute this and prevent it. Future research to mitigate this risk might build on work such as \citet{piazza2024power}, where `power regularisation' of information flow is proposed to limit the adverse effects of communication by misaligned agents. 

Similarly, for the non-tabular MFG setting in \citet{benjamin2024networked}, the work considers the robustness of the population to scenarios when agents fail to update their policies by the time they are required to communicate. Similar studies in our setting are the matter of future work. Likewise, our experiments only use one round of communication within each iteration to show the benefit that even this can have on learning speed over the independent case, but it may be informative to study the effect of communication costs further.

Alg. \ref{alg:mean_field_estimation_specific} assumes that if a state $s'$ is connected to $s$ on the visibility graph $\mathcal{G}^{vis}_{t}$, an agent in $s$ is able to \textit{accurately} count all the agents in $s'$, i.e. it either counts the exact total or cannot observe the state at all. We assume this for simplicity but it is not inherently the case, since a real-world agent may have only noisy observations even of others located nearby, due to imperfect sensors. We suggest two ways to deal with this. Firstly, if agents share unique IDs as in Alg. \ref{alg:mean_field_estimation_general}, then when communicating their vectors of collected IDs with each other via $\mathcal{G}^{comm}_{t}$, agents would gain the most accurate picture possible of all the agents that have been observed in a given state. However, as we note above, there are various reasons why sharing IDs might be undesirable, including privacy and scalability. If instead only counts are taken, and if the noise on each agents' count is assumed to be independent and, for example, subject to a Gaussian distribution, the algorithm can easily be updated such that communicating agents compute averages of their local and received counts to improve their accuracy, rather than simply using communication to fill in counts for previously unseen states. (Note that we can also consider the original case without noise to involve averaging, since averaging identical values equates to using the original value). 
Since the algorithm is intended to aid in local estimation of the mean-field distribution, which is inherently approximate due to the uniform method for distributing the uncounted agents, we are not concerned with reaching exact consensus between agents on the communicated counts, so we do not require repeated averaging to ensure asymptotic convergence. 

We may wish to consider more sophisticated methods for distributing the uncounted agents across states, in place of the current uniform distribution. Such choices may be domain-specific based on knowledge of a particular environment. For example, one might use the counts to perform Bayesian updates on a specific prior, where this prior may relate to the estimated mean-field distribution at the previous time step $t-1$. If agents seek to learn to predict the \textit{evolution} of the mean field based on their own policy or by learning a model, the Bayesian prior may also be based on forward prediction from the estimated mean-field distribution at $t-1$. Future work lies in conducting experiments in all of these more general settings. 

\citet{perrin2022generalization} notes that in grid-world settings such as those in our experiments, passing the (estimated or true global) mean-field distribution as a flat vector to the Q-network ignores the geometric structure of the problem. They therefore propose to create an embedding of the distribution by first passing the vector to a convolutional neural network, essentially treating the categorical distribution as an image. This technique is also followed in \citet{wu2024populationaware} (for their additional experiments, but not in the main body of their paper). As future work, we can test whether such a method improves the performance of our algorithms.


\end{document}